\documentclass[12pt]{article}
\pdfoutput=1
\usepackage[utf8]{inputenc}
\usepackage[english]{babel}
\usepackage[scale=0.75]{geometry}
\usepackage{csquotes}
\usepackage{amsmath}
\usepackage{amsfonts}
\usepackage{amssymb}
\usepackage{hyperref}
\usepackage{ifthen}
\usepackage{tikz}
\usetikzlibrary{calc}

\usepackage[backend=biber,style=numeric,sorting=ynt]{biblatex}
\usepackage{amsthm}

\theoremstyle{plain}
\newtheorem*{theorem}{Theorem}
\newtheorem{corollary}{Corollary}
\newtheorem*{lemma}{Lemma}
\newtheorem{lemma2}{Lemma}
\newtheorem{claim}{Claim}

\theoremstyle{definition}
\newtheorem{definition}{Definition}
\newtheorem*{definition2}{Definition}

\newcommand{\spinup}{|\!\uparrow\rangle}
\newcommand{\spindown}{|\!\downarrow\rangle}

\newcommand{\bra}[1]{\langle#1\vert}
\newcommand{\ket}[1]{\vert#1\rangle}

\newcommand{\twovect}[2]{({#1}\;,\;{#2})}

\newcommand\tri{} 
\def\tri[#1]#2(#3,#4)#5(#6){
  \draw[#1] (#3,{#4*sqrt(3)}) -- +(-#6,{-sqrt(3)*#6}) -- +(#6,{-sqrt(3)*#6}) -- cycle;
}

\title{Quantum Pascal's Triangle and Sierpinski's carpet}

\author{
Tom Bannink\thanks{QuSoft and CWI Amsterdam, Science Park 123, 1098 XG Amsterdam, The Netherlands} \\ \small{bannink@cwi.nl}
\and Harry Buhrman\stepcounter{footnote}\footnotemark[1]\;\;\thanks{University of Amsterdam, Science Park 904, 1098 XH Amsterdam, The Netherlands} \\ \small{buhrman@cwi.nl}}

\date{August 2017}

\begin{document}
\maketitle
\begin{abstract}
In this paper we consider a quantum version of Pascal's triangle. Pascal's triangle is a well-known triangular array of numbers and when these numbers are plotted modulo 2, a fractal known as the Sierpinski triangle appears. We first prove the appearance of more general fractals when Pascal's triangle is considered modulo prime powers. The numbers in Pascal's triangle can be obtained by scaling the probabilities of the simple symmetric random walk on the line. In this paper we consider a quantum version of Pascal's triangle by replacing the random walk by the quantum walk known as the Hadamard walk. We show that when the amplitudes of the Hadamard walk are scaled to become integers and plotted modulo three, a fractal known as the Sierpinski carpet emerges and we provide a proof of this using Lucas's theorem. We furthermore give a general class of quantum walks for which this phenomenon occurs.
\end{abstract}

\section{Introduction}
Pascals's triangle, shown in Figure \ref{fig:pascal1}, exhibits many interesting properties one of which is the appearance of a fractal when the numbers are considered modulo a prime $p$ \cite{Wolfram1984,Stewart1995}. This is shown in Figure \ref{fig:pascal3}, and for $p=2$ the fractal is known as the Sierpinski triangle or Sierpinski gasket. One way of obtaining the numbers in Pascal's triangle is through a random walk on a line as will be explained in Section \ref{sec:pascal}. This paper explores the results of considering a similar triangle of numbers that is obtained when the 1-dimensional random walk is replaced by a 1-dimensional quantum walk. This also yields the Sierpinski triangle when the probabilities associated to the quantum walk are considered modulo 2, but more interestingly one can find another fractal known as the Sierpinski carpet hidden in the amplitudes modulo 3 which is not present in Pascal's triangle. When these quantum walk numbers are plotted modulo $p$, more general fractals appear.

\begin{table}
    \makebox[\textwidth][c]{
    \begin{tabular}{r|l|l|l}
        ~ & mod 2 & mod 3 & mod $p$ \\
        \hline
        Pascal's triangle       & Triangle (2) & Triangle (3) & Triangle ($p$) \\
        Hadamard walk           & Triangle (2) & Carpet       & See Figure \ref{fig:generalfractals} \\
        General quantum walk    & Triangle (2) & Carpet or Triangle (3) & See Figure \ref{fig:generalfractals}
    \end{tabular}
    }
    \caption{Summary of the fractals that result from considering various sets of numbers modulo a prime. Triangle ($p$) refers to the version of Sierpinski triangle where $p(p+1)/2$ copies of the triangle are found in every recursion level. See Figure \ref{fig:pascal2} for $p\in\{2,3,5,7\}$. Carpet refers to the Sierpinski carpet as shown in Figure \ref{fig:carpet1}.}
    \label{tab:summary}
\end{table}

This paper starts with Pascal's triangle and shows how it is related to the Sierpinski triangle when the numbers are taken modulo a prime. We then provide a proof of the appearance of a more general version of the Sierpinski triangle when instead we take prime \emph{powers}. Then, quantum walks are introduced with an emphasis on a walk that is commonly known as the Hadamard walk. We derive an expression for the probabilities of these walks and then the appearance of both the Sierpinski triangle and Sierpinski carpet is shown as well as some other properties.
Table \ref{tab:summary} provides a summarising overview of the different fractals that are obtained from these different sources.

\section{Pascal's triangle} \label{sec:pascal}
Pascal's triangle is a set of integers arranged in a triangle, where the $k$'th value in the $n$'th row (both $n$ and $k$ start at zero) is given by the binomial coefficient $\binom{n}{k}$. It is shown in Figure \ref{fig:pascal1}, and can also be constructed by using $\binom{n}{k}=\binom{n-1}{k-1}+\binom{n-1}{k}$, i.e. every number is the sum of its two neighbours in the row above. Alternatively it can be thought of as `scaled' probabilities of a random walk on $\mathbb{Z}$, in the following way. Assume the random walk starts in the origin, and goes left or right with probability $\frac{1}{2}$. The probability of being at location $l\in\mathbb{Z}$ after $n$ steps, with $-n\leq l \leq n$ is given by $\frac{1}{2^n}\binom{n}{(n+l)/2}$ if $n+l$ is even and 0 if $n+l$ is odd. When considering only the $n+1$ non-zero probabilities after $n$ steps, the $k$'th non-zero value corresponds to position $l=-n+2k$ of the line, where $0\leq k\leq n$. The $k$'th non-zero probability is given by $\frac{1}{2^n}\binom{n}{k}$. Removing the factor $\frac{1}{2^n}$ yields the integer numbers in Pascal's triangle.
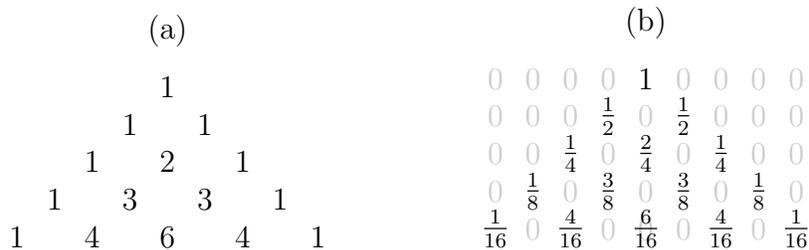
\begin{figure}
    \begin{center}
        \begin{tikzpicture}
            \def\scale{0.5};
            
            \node at ( 0*\scale,1.5*\scale) {(a)};
            
            \node at ( 0*\scale,  0*\scale) {$1$};
            \node at (-1*\scale, -1*\scale) {$1$};
            \node at (+1*\scale, -1*\scale) {$1$};
            
            \node at (-2*\scale, -2*\scale) {$1$};
            \node at ( 0*\scale, -2*\scale) {$2$};
            \node at (+2*\scale, -2*\scale) {$1$};
            
            \node at (-3*\scale, -3*\scale) {$1$};
            \node at (-1*\scale, -3*\scale) {$3$};
            \node at (+1*\scale, -3*\scale) {$3$};
            \node at (+3*\scale, -3*\scale) {$1$};
            
            \node at (-4*\scale, -4*\scale) {$1$};
            \node at (-2*\scale, -4*\scale) {$4$};
            \node at ( 0*\scale, -4*\scale) {$6$};
            \node at (+2*\scale, -4*\scale) {$4$};
            \node at (+4*\scale, -4*\scale) {$1$};           
        \end{tikzpicture}
        \qquad\qquad
        \begin{tikzpicture}
            \def\scale{0.5};
            
            \node at ( 0*\scale,1.5*\scale) {(b)};
            
            \node at ( 0*\scale,  0*\scale) {$1$};
            \node at (-1*\scale, -1*\scale) {$\frac{1}{2}$};
            \node at (+1*\scale, -1*\scale) {$\frac{1}{2}$};
            
            \node at (-2*\scale, -2*\scale) {$\frac{1}{4}$};
            \node at ( 0*\scale, -2*\scale) {$\frac{2}{4}$};
            \node at (+2*\scale, -2*\scale) {$\frac{1}{4}$};
            
            \node at (-3*\scale, -3*\scale) {$\frac{1}{8}$};
            \node at (-1*\scale, -3*\scale) {$\frac{3}{8}$};
            \node at (+1*\scale, -3*\scale) {$\frac{3}{8}$};
            \node at (+3*\scale, -3*\scale) {$\frac{1}{8}$};
            
            \node at (-4*\scale, -4*\scale) {$\frac{1}{16}$};
            \node at (-2*\scale, -4*\scale) {$\frac{4}{16}$};
            \node at ( 0*\scale, -4*\scale) {$\frac{6}{16}$};
            \node at (+2*\scale, -4*\scale) {$\frac{4}{16}$};
            \node at (+4*\scale, -4*\scale) {$\frac{1}{16}$};           
            
            \foreach \y in {0,...,4} {
            \foreach \x in {-4,...,4} {
                \ifthenelse{\x > \y \OR \x < -\y}{
                \node[opacity=0.2] at (\x*\scale, -\y*\scale) {$0$};
                }{}
            }
            }
            \node[opacity=0.2] at ( 0*\scale, -1*\scale) {$0$};
            \node[opacity=0.2] at (-1*\scale, -2*\scale) {$0$};
            \node[opacity=0.2] at (+1*\scale, -2*\scale) {$0$};
            \node[opacity=0.2] at (-2*\scale, -3*\scale) {$0$};
            \node[opacity=0.2] at ( 0*\scale, -3*\scale) {$0$};
            \node[opacity=0.2] at (+2*\scale, -3*\scale) {$0$};
            \node[opacity=0.2] at (-3*\scale, -4*\scale) {$0$};
            \node[opacity=0.2] at (-1*\scale, -4*\scale) {$0$};
            \node[opacity=0.2] at ( 0*\scale, -4*\scale) {$0$};
            \node[opacity=0.2] at (+1*\scale, -4*\scale) {$0$};
            \node[opacity=0.2] at (+3*\scale, -4*\scale) {$0$};
        \end{tikzpicture}
    \end{center}
    \caption{\label{fig:pascal1}The top five rows of Pascal's triangle (a) and the probabilities of the first 5 steps of a simple random walk (b). The probabilities equal to zero in (b) are in light-grey for clarity.}
\end{figure}
\begin{figure}
    \begin{center}
        \makebox[\textwidth][c]{\includegraphics[scale=1.4]{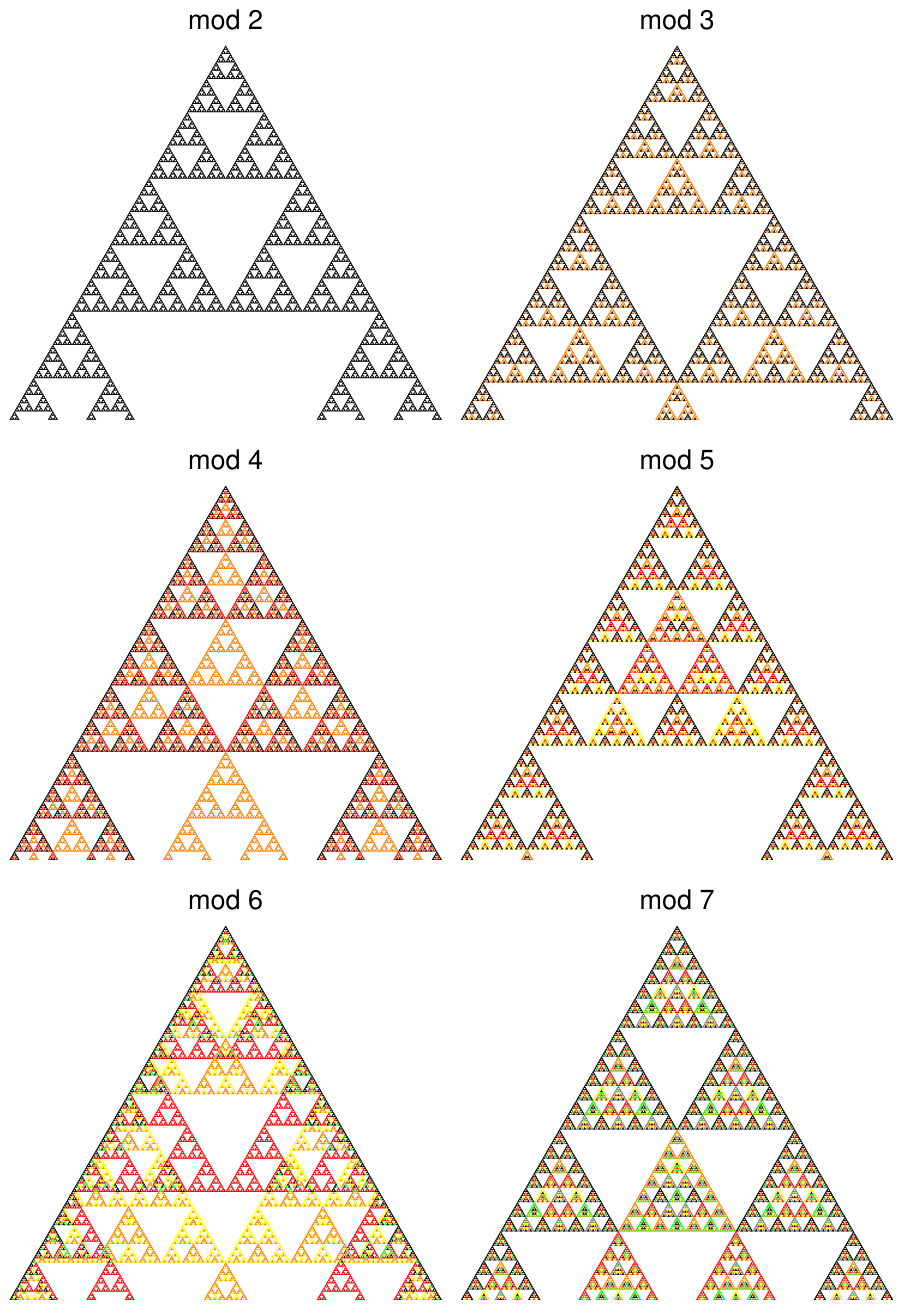}}
    \end{center}
    \caption{\label{fig:pascal2} \label{fig:pascal3}The first 180 of rows of Pascal's triangle shown modulo $n$ where $n\in\{2,3,4,5,6,7\}$. If a value was zero modulo $n$ it is coloured white, otherwise it is given a different colour.}
\end{figure}

\subsection{Pascal's triangle modulo two}
Pascal's triangle has many interesting properties and one interesting feature comes from considering all numbers modulo two \cite{Wolfram1984}. This `binary' triangle is shown in Figure \ref{fig:pascal2} where black and white pixels are used to represent the values modulo two. The figure that appears looks very much like the Sierpinski triangle (also known as the Sierpinski gasket). Indeed, in the limit of an infinite number of rows, Pascal's triangle modulo two is the Sierpinski gasket.

\begin{figure}
    \begin{center}
        \includegraphics[width=0.8\textwidth]{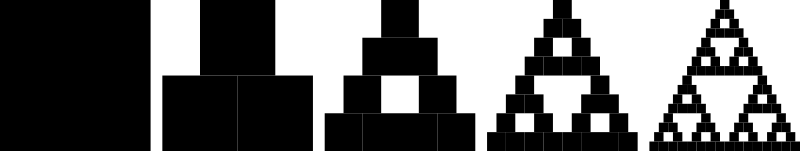}
    \end{center}
    \caption{\label{fig:triangle_construction} First few iterations of constructing the Sierpinski triangle. In each iteration the previous shape is shrunk to half its size and three copies are put in the corners of a triangle so that the shapes are touching.}
\end{figure}
To prove that Pascal's triangle modulo two converges to the Sierpinski triangle, a definition of the Sierpinski triangle is needed. There are different ways to construct it, and one of them is by shrinking and duplication \cite{Barnsley2003}. This process is shown in Figure \ref{fig:triangle_construction}. Start with any shape (a closed bounded region) in the plane, like shown in the first image. Shrink the shape to half its size (both height and width) and make three copies of it. Place these copies in the corners of an equilateral triangle such that the shapes touch as shown in the second image. Repeat this with the new shape. The rightmost image shows the result after four iterations, and after an infinite number of iterations one obtains the fractal.

\begin{figure}
    \begin{center}
        \includegraphics[scale=0.8]{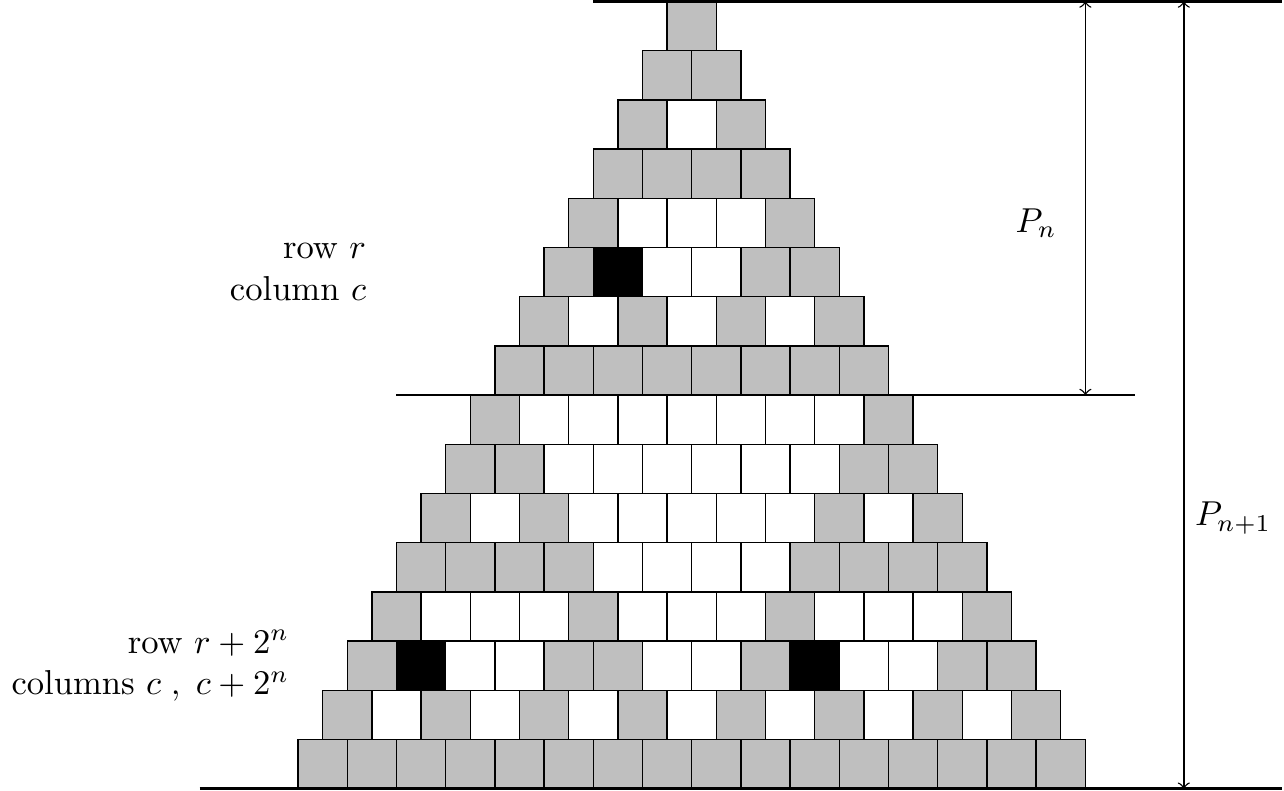}
    \end{center}
    \caption{\label{fig:pascal_row_col} The first $2^{n+1}$ rows of Pascal's triangle modulo two ($P_{n+1}$) contain three copies of the first $2^n$ rows of the triangle ($P_n$). A point at row $r$ and column $c$ in the original triangle is copied to row $r+2^n$ and columns $c$ and $c+2^n$.}
\end{figure}
To show that Pascal's triangle modulo two converges to the fractal, one can show that for all $n$, the shape given by the first $2^{n+1}$  rows contains exactly three copies of the first $2^n$ rows, and nothing more than those copies, i.e. with `white' in-between, see Figure \ref{fig:pascal_row_col}. If we index Pascal's triangle by row $r$ and column $c$, with $r\geq 0$ and $0\leq c \leq r$ then we need to show that for $0\leq r < 2^n$ and $0\leq c \leq r$,
\begin{align}
\label{eq:pascal1}
    \binom{r}{c} \equiv \binom{r+2^n}{c} \equiv \binom{r+2^n}{c+2^n} \mod 2.
\end{align}
Here $\binom{r+2^n}{c}$ corresponds to a `pixel' in the copy at the bottom left, and $\binom{r+2^n}{c+2^n}$ corresponds to the copy at the bottom right, as shown in Figure \ref{fig:pascal_row_col}.
Furthermore, we need to show that there is nothing more than these copies, i.e. the triangular area in-between the copies is empty. This means that the values in the corresponding region are even, and it is sufficient to show that the top row of the region is even with odd endpoints. The top row of this region is at $r=2^n$, and the endpoints of the row, $c=0$ and $c=2^n$ are odd since $\binom{r}{0}=\binom{r}{r}=1$. When the remainder of this row is even, i.e. for $1\leq c \leq 2^n - 1$,
\begin{align}
\label{eq:pascal2}
    \binom{2^n}{c} &\equiv 0 \mod 2,
\end{align}
then the complete region will be even. This is because a value in Pascal's triangle is the sum of its two neighbours in the row above. Adding two even numbers results in an even number, so the next row ($r=2^{n}+1$) will also have even values except near the boundary where an even number is added to an odd number. In particular, at every row the size of the `even part' in the middle decreases by one, resulting in an empty up-side-down triangle ending at $r=2^n + 2^n - 1$ where the two endpoints have joined. By proving that the top row is even it follows that all entries in the up-side-down triangle are even. We conclude that it is sufficient to prove (\ref{eq:pascal1}) and (\ref{eq:pascal2}). For this, one can use Lucas's theorem.

It is convenient to introduce notation for representing a number by it's base-$p$ digits for a prime $p$. We will write
\begin{align*}
    n = [n_m n_{m-1}\cdots n_0]_p = \sum_{j=0}^{m} n_j \; p^j \qquad \text{with } 0 \leq n_i < p \text{ for each } i,
\end{align*}
where the $n_i$ are the base-$p$ digits of $n$

\begin{theorem}[Luc1878\nocite{Lucas78}]
Let $p$ be prime and $n,k$ non-negative integers. Let $n=[n_m n_{m-1}\cdots n_0]_p$ and $k=[k_m k_{m-1}\cdots k_0]_p$. Then
\begin{align}
    \label{eq:lucas}
    \binom{n}{k} \equiv \binom{n_m}{k_m}\binom{n_{m-1}}{k_{m-1}} \cdots \binom{n_0}{k_0} \mod p,
\end{align}
where we define $\binom{n}{k}=0$ if $k>n$.
\end{theorem}

There are many extensions and generalisations of Lucas's theorem \cite{Mestrovic2014}, that include versions for prime powers or similar congruences for generalised binomial coefficients, but they are not needed here.

\begin{corollary} \label{cor:zero}
For any prime $p$,
\begin{align*}
    \binom{n}{k} \equiv 0 \mod p \quad \iff \quad \exists i \; : \; k_i > n_i
\end{align*}
\end{corollary}
\begin{proof}
If there is an $i$ such that $k_i > n_i$ then $\binom{n_i}{k_i}=0$ and by Lucas's theorem, $\binom{n}{k}\equiv 0 \mod p$. Conversely, if $k_i \leq n_i$ for all $i$ then $\binom{n_i}{k_i}=\frac{n_i!}{k_i!(n_i-k_i)!}$. Since $n_i <p$ and $p$ is prime, we have that $p$ is not a factor of $n_i!$ and also not a factor of $\frac{n_i!}{k_i!(n_i-k_i)!}$. Since $p$ is not a divisor of $\binom{n_i}{k_i}$ and $p$ is prime, it also does not divide the product $\prod_i \binom{n_i}{k_i}$ which concludes the proof.
\end{proof}
The following corollary considers adding extra digits to $n$ and $k$ and is also known as Anton's Lemma:
\begin{corollary} \label{cor:addpower}
If $n,k < p^{m}$ and then for all $l,q\geq 0$,
\begin{align}
    \binom{l\cdot p^m + n}{q\cdot p^m + k} \equiv \binom{l}{q}\binom{n}{k} \mod p. \label{eq:trick2}
\end{align}
\end{corollary}
\begin{proof}
When $n,k<p^m$ then $n_i=k_i=0$ for $i\geq m$, so that $l\cdot p^m + n = [l_M l_{M-1}\cdots l_0 n_{m-1} n_{m-2}\cdots n_0]_p$ and $q\cdot p^m + k = [q_M q_{M-1}\cdots q_0 k_{m-1} k_{m-2}\cdots k_0]_p$. Therefore, by Lucas's theorem
\begin{align*}
    \binom{l\cdot p^{m}+n}{q\cdot p^m + k} &\equiv \binom{l_M}{q_M}\binom{l_{M-1}}{q_{M-1}}\cdots \binom{l_0}{q_0}\binom{n_{m-1}}{k_{m-1}} \cdots \binom{n_0}{k_0} \mod p\\
                                           &\equiv \binom{l}{q}\binom{n}{k} \mod p.
\end{align*}
\end{proof}
Note that vice versa, Corollary \ref{cor:addpower} implies Lucas's theorem by induction on the number of digits.

We can now prove (\ref{eq:pascal1}) and (\ref{eq:pascal2}). By Corollary \ref{cor:addpower} for $p=2$ we have
\begin{align*}
    \binom{r+2^n}{c} \equiv \binom{1}{0}\binom{r}{c} \mod 2,\\
    \binom{r+2^n}{c+2^n} \equiv \binom{1}{1}\binom{r}{c} \mod 2,
\end{align*}
so (\ref{eq:pascal1}) follows from $\binom{1}{0}=\binom{1}{1}=1$. To show (\ref{eq:pascal2}), note that since $1\leq c \leq 2^n -1$ there is a digit $c_i$ that is nonzero for $i<n$, whereas all digits of $r=2^n$ are zero except for the $r_{n}$, so (\ref{eq:pascal2}) follows form Corollary \ref{cor:zero}.

\subsection{Pascals triangle modulo general $n$}

In a similar fashion one can consider Pascal's triangle modulo general $n$. Figure \ref{fig:pascal3} shows this for $n\in\{2,3,4,5,6,7\}$. We can distinguish cases for primes, prime powers and other numbers.

\subsubsection{Pascals triangle modulo a prime}

When $n$ is a prime (2,3,5,7 in the figure), one obtains a generalisation of Sierpinski's triangle. This generalisation for primes $p$ can also be constructed using the shrinking and duplication method. When using the shrinking and duplication construction, one can start with an arbitrary shape, shrink it and create $p(p+1)/2$ copies. These copies have to be arranged into a larger triangle where all the copies are touching. The proof is a generalisation of the one given in the previous section. One has to show that
\begin{align*}
    \binom{r}{c} \equiv \binom{r+l\cdot p^n}{c+q\cdot p^n} \mod p,
\end{align*}
where $0\leq l < p$ and $0\leq q \leq l$. Each value of $(l,q)$ corresponds to one of the $p(p+1)/2$ copies. This equivalence follows from Corollary \ref{cor:addpower} and $\binom{l}{q}\not\equiv 0 \mod p$. The empty up-side-down triangles correspond to $\binom{r+l\cdot p^m}{c+q\cdot p^m}$ but where the range of $c$ is now $r<c<p^m$ as opposed to $0\leq c \leq r$. This case is also included in Corollary \ref{cor:addpower}, and as $\binom{r}{c}=0$ for $r<c$ this finishes the proof.

\subsubsection{Pascals triangle modulo a composite number}

When $n$ is composite (mod 6 in the figure) then the resulting shape is the union of the shapes obtained of its factors (prime powers) albeit with different colours. For example, at $n=6$, shown in Figure \ref{fig:pascal3} one can see the union of the shapes of $p=2$ and $p=3$. This is simply because when $n=p_1^{k_1}\cdots p_m^{k_m}$ then $x\equiv 0 \mod n$ if and only if for all $i$ : $x \equiv 0 \mod p_i^{k_i}$.

\subsubsection{Pascals triangle modulo a prime power}

When $n$ is a prime power then the pattern becomes slightly more complicated. One can see in Figure \ref{fig:pascal3} that for $n=4$, the image is the same as for $n=2$ but with extra triangles in the places that used to be empty. When one would consider $n=8$, this idea is repeated and the holes in the $n=4$ shape are filled with additional triangles.
\begin{figure}
    \makebox[\textwidth][c]{\includegraphics[scale=1.1]{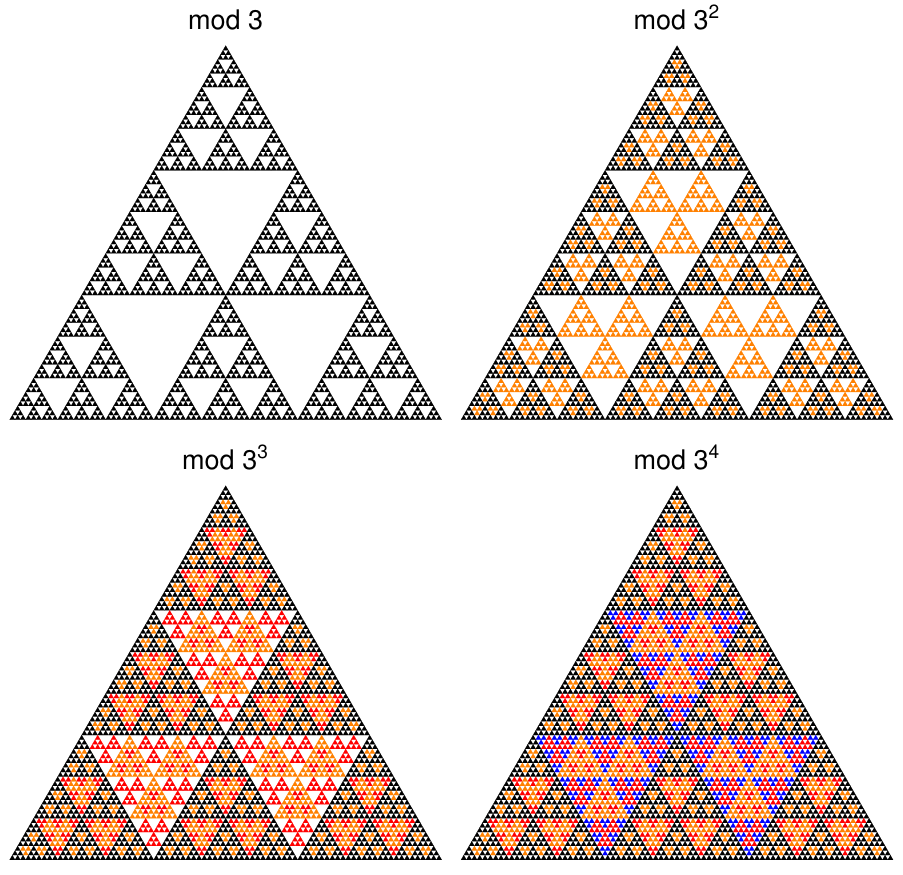}}
    \caption{\label{fig:pascal3powers} Pascal's triangle plotted modulo powers of $3$. The colours give the $3$-adic valuation $\nu_3(\binom{r}{c})$, where black is $0$, orange is $1$, red is $2$ and blue is $3$.}
\end{figure}
Figure \ref{fig:pascal3powers} shows what happens when the triangle is plotted modulo powers of $3$. From the Figure we can conjecture the general pattern: Start with the `mod $p$ triangle' and in every empty region, add $p(p-1)/2$ copies of the `mod $p$ triangle'. This yields the pattern for the `mod $p^2$ triangle'. To go to the `mod $p^3$ triangle', again add $p(p-1)/2$ copies of the `mod $p$ triangle' to the empty regions of the `mod $p^2$ triangle'.
This process can be iterated to find the shape corresponding to $p^k$ for any $k$. To make this statement more concrete, we need the following definition.
\begin{definition2}
    The $p$-adic valuation $\nu_p(n)$ of a number $n$ is the largest power of $p$ that divides $n$.
\end{definition2}

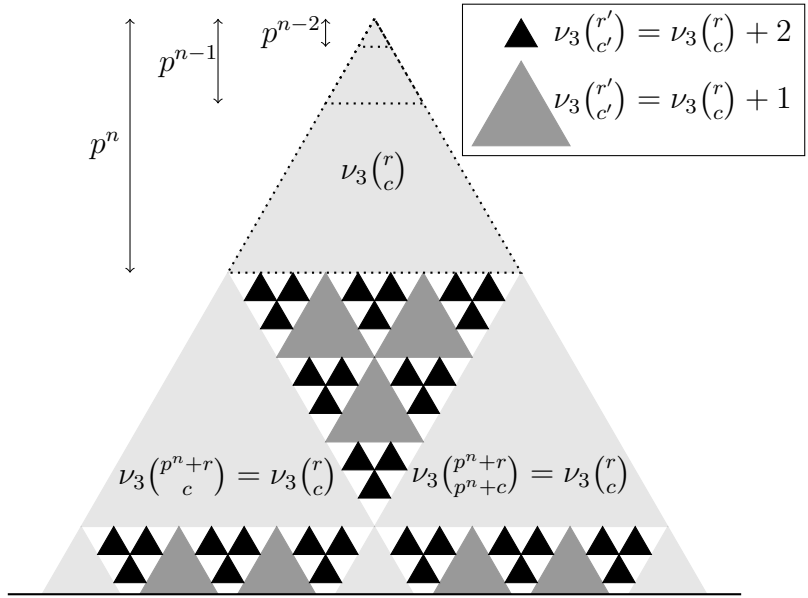
\begin{figure}
    \makebox[\textwidth][c]{
\begin{tikzpicture}[scale=0.65]
    \tri[fill,black,opacity=0.1] (0,0) (3);
    \tri[fill,black,opacity=0.1] (-3,-3) (3);
    \tri[fill,black,opacity=0.1] ( 3,-3) (3);
    
    \tri[fill,black,opacity=0.4] (-1,-3) (1);
    \tri[fill,black,opacity=0.4] ( 1,-3) (1);
    \tri[fill,black,opacity=0.4] ( 0,-4) (1);
    
    \tri[fill,black] (-2.333, -3) (0.33);
    \tri[fill,black] (-1.666, -3) (0.33);
    \tri[fill,black] (-2, -3.333) (0.33);
    
    \tri[fill,black] ( 2.333, -3) (0.33);
    \tri[fill,black] ( 1.666, -3) (0.33);
    \tri[fill,black] ( 2, -3.333) (0.33);
    
    \tri[fill,black] (-0.333, -3) (0.33);
    \tri[fill,black] ( 0.333, -3) (0.33);
    \tri[fill,black] ( 0, -3.333) (0.33);   
    
    \tri[fill,black] (-0.333, -5) (0.33);
    \tri[fill,black] ( 0.333, -5) (0.33);
    \tri[fill,black] ( 0, -5.333) (0.33);      
    
    \tri[fill,black] (-1.333, -4) (0.33);
    \tri[fill,black] (-0.666, -4) (0.33);
    \tri[fill,black] (-1, -4.333) (0.33);      
    
    \tri[fill,black] ( 1.333, -4) (0.33);
    \tri[fill,black] ( 0.666, -4) (0.33);
    \tri[fill,black] ( 1, -4.333) (0.33);      
    
    \tri[fill,black] (-5.333, -6) (0.33);
    \tri[fill,black] (-5, -6.333) (0.33);
    \tri[fill,black] (-4.666, -6) (0.33);
    \tri[fill,black] (-3.333, -6) (0.33);
    \tri[fill,black] (-3, -6.333) (0.33);
    \tri[fill,black] (-2.666, -6) (0.33);
    \tri[fill,black] (-1.333, -6) (0.33);
    \tri[fill,black] (-1, -6.333) (0.33);
    \tri[fill,black] (-0.666, -6) (0.33);
    \tri[fill,black] ( 5.333, -6) (0.33);
    \tri[fill,black] ( 5, -6.333) (0.33);
    \tri[fill,black] ( 4.666, -6) (0.33);
    \tri[fill,black] ( 3.333, -6) (0.33);
    \tri[fill,black] ( 3, -6.333) (0.33);
    \tri[fill,black] ( 2.666, -6) (0.33);
    \tri[fill,black] ( 1.333, -6) (0.33);
    \tri[fill,black] ( 1, -6.333) (0.33);
    \tri[fill,black] ( 0.666, -6) (0.33);
    \tri[fill,black,opacity=0.4] (-4,-6) (0.80);
    \tri[fill,black,opacity=0.4] (-2,-6) (0.80);
    \tri[fill,black,opacity=0.4] ( 2,-6) (0.80);
    \tri[fill,black,opacity=0.4] ( 4,-6) (0.80);
    \tri[fill,black,opacity=0.1] (-6,-6) (0.80);
    \tri[fill,black,opacity=0.1] ( 0,-6) (0.80);
    \tri[fill,black,opacity=0.1] ( 6,-6) (0.80);
    \draw[thick] (-7.5,{-6.80*sqrt(3)}) -- +(15,0);
    
    \tri[dotted,thick] (0,0) (0.33);
    \tri[dotted,thick] (0,0) (1);
    \tri[dotted,thick] (0,0) (3);
    
    \draw ( 0,{-1.8*sqrt(3)}) node {$\nu_3\binom{r}{c}$};
    \draw (-3,{-5.4*sqrt(3)}) node {$\nu_3\binom{p^n + r}{c} = \nu_3\binom{r}{c}$};
    \draw ( 3,{-5.4*sqrt(3)}) node {$\nu_3\binom{p^n + r}{p^n + c} = \nu_3\binom{r}{c}$};
    
    
    \draw (1.8,0.3) rectangle (8.8,-2.8);
    \tri[fill,black] (3,-0.0) (0.33);
    \draw (6.1,{-0.0*sqrt(3)-0.25}) node {$\nu_3\binom{r'}{c'} = \nu_3\binom{r}{c}+2$};
    \tri[fill,black,opacity=0.4] (3,-0.5) (1.0);
    \draw (6.1,{-0.8*sqrt(3)-0.25}) node {$\nu_3\binom{r'}{c'} = \nu_3\binom{r}{c}+1$};    
    
    \draw[<->] (-5.0,0) -- (-5.0,{-3*sqrt(3)});
    \draw (-5.5,{-1.5*sqrt(3)}) node {$p^n$};
    
    \draw[<->] (-3.2,0) -- (-3.2,{-1*sqrt(3)});
    \draw (-3.8,{-0.5*sqrt(3)}) node {$p^{n-1}$};
    
    \draw[<->] (-1,0) -- (-1,{-0.33*sqrt(3)});
    \draw (-1.7,{-0.16*sqrt(3)}) node {$p^{n-2}$};
\end{tikzpicture}
    }
    \caption{\label{fig:valuations}Schematic overview of the statements we prove for the `mod $p^k$ triangle', here shown for $p=3$. The triangles at the top with dashed lines correspond to values $\binom{r}{c}$ of Pascal's triangle. We show that in the other size-$p^n$ triangles, the $p$-adic valuation of the corresponding numbers ($r\to l\cdot p^n + r$ etc) is the same. For the smaller triangles of size $p^{n-k}$, the $p$-adic valuation is $k$ higher.}
\end{figure}
The statements that we want to prove are most easily explained with a picture, shown in Figure \ref{fig:valuations}. We will show that at each recursion level of the triangle, the $p$-adic valuation of the numbers $\binom{r'}{c'}$ in a copy (meaning $r'=l p^n + r$ and $c' = q p^n + c$) is the same as that of the corresponding number $\binom{r}{c}$ in the original region. Furthermore, we show that the regions that were empty in the `mod $p$ triangle' have particular $p$-adic valuations that are 1 or 2 or $k$ higher than the original, as depicten in Figure \ref{fig:valuations}. Every time you fill a previously-empty region with triangles, the $p$-adic valuation increases by one.

We will now prove this by using Kummer's theorem.

~\\ 

\begin{theorem}[Kum1852\nocite{kummer}]
Let $p$ be prime and $n,k$ non-negative integers, $n\geq k$. Then the $p$-adic valuation $\nu_p(\binom{n}{k})$ of $\binom{n}{k}$ is equal to the number of ``carries'' when $k$ and $n-k$ are added in base-$p$ arithmetic.
\end{theorem}
One way to find the number of carries that occur when $k$ is added to $n-k$ in base-$p$ is by considering the base-$p$ digits of $n$ and $k$, defining $c^{n,k}_{-1}=0$ and
\begin{align*}
    c^{n,k}_{i} = \begin{cases} 1 & n_i < k_i\\ 0 & n_i > k_i\\ c^{n,k}_{i-1} & n_i = k_i \end{cases}.
\end{align*}
The number of carries is then equal to $\sum_{i\geq 0} c^{n,k}_{i}$. Kummer's theorem can therefore be written as $\nu_p(\binom{n}{k}) = \sum_{i\geq 0} c^{n,k}_{i}$.

The following claim shows what happens to $\nu_p(\binom{n}{k})$ when a digit $l$ is added to $n$ and a digit $q$ is added to $k$:
\begin{claim}\label{claim:valuationdigits}
    Let $p$ be prime and $n,k,q,l,m$ non-negative integers with $0\leq k \leq n < p^m$ and $0 \leq q \leq l < p$. Then
    \begin{align*}
        \nu_p(\binom{l\cdot p^m + n}{q\cdot p^m + k}) = \nu_p(\binom{n}{k})
    \end{align*}
\end{claim}
\begin{proof}
    Define $n'=l\cdot p^m +n$ and $k'=q\cdot p^m + k$.
    Note that $n,k$ have at most $m-1$ digits when expressed in base $p$ and $l,q$ are the $m$-th digits of $n'$ and $k'$. Note that we have $c^{n,k}_{i}=c^{n',k'}_{i}$ for $i<m$ since the first $m-1$ digits are the same. For the $m$-th digits we have
    \begin{align*}
        c^{n',k'}_{m}=\begin{cases} 1 & l < q\\ 0 & l > q\\ c^{n,k}_{m-1} & l = q \end{cases} .
    \end{align*}
    By Kummer's theorem the difference between $\nu_p(\binom{n'}{k'})$ and $\nu_p(\binom{n}{k})$ is equal to $c^{n',k'}_{m}$, so it remains to show that $c^{n',k'}_{m}=0$. By assumption we know $q\leq l$ and if $q<l$ we have $c^{n',k'}_{m} = 0$ by definition. Consider the case $l=q$ where we have $c^{n',k'}_{m} = c^{n,k}_{m-1}$. If $n=k$ then all the $c^{n,k}_{i}$ are zero so we are done. If $n \neq k$ then the consider the most significant digit where $n$ and $k$ differ, i.e. take the highest $i$ for which $n_i \neq k_i$ and call it $i^*$. Since $k<n$ by assumption, it must be true that $k_{i^*} < n_{i^*}$ and therefore $c^{n,k}_{i^*} = 0$. For all $i>i_*$ we have $n_i=k_i$ so $c^{n,k}_{i}=c^{n,k}_{i-1}$. So $c^{n',k'}_{m} = c^{n,k}_{i_*} = 0$.
\end{proof}
\begin{claim}\label{claim:valuationmod}
    If $\nu_p(n) = \nu_p(m)$ then for any $k$ $$n\equiv 0 \mod p^k \iff m\equiv 0 \mod p^k.$$
\end{claim}
\begin{proof}
    It follows from the fact that $n\equiv 0 \mod p^k$ if and only if $\nu_p(n) \geq k$.
\end{proof}

Consider Figure \ref{fig:pascal3powers}.
The size of the recursion levels in the `mod $p^k$ triangle' is the same as in the mod $p$ triangle, meaning powers of $p$ and not powers of $p^k$. Repeating what we did before for the mod $p$ triangle, we can see that at recursion level $n$, the ``copies'' and ``empty regions'' correspond to the following binomial coefficients of Pascal's triangle:
\begin{align*}
    \binom{l\cdot p^n+r}{q\cdot p^n + c} & &
    \begin{array}{rl}
        \text{``copies''} \to & 0 \leq q \leq l < p \; , \; 0 \leq c \leq r < p^n \\
        \text{``empty''}  \to & 0 \leq q < l < p \; , \; 0 \leq r < c < p^n 
    \end{array}
    .
\end{align*}
Here $n$ is the recursion level and $(l,q)$ index the different copies or empty regions whereas $r$ and $c$ index points within those regions. By claim \ref{claim:valuationdigits}, the $p$-adic valuation of a number $\binom{r'}{c'}$ in the copy is the same as that of the corresponding original number $\binom{r}{c}$. By Claim \ref{claim:valuationmod}, we have now shown that for any $k$, the copies in the `mod $p^k$ triangle' are indeed all the same.

What is left to show is how the empty regions of the `mod $p$ triangle' are filled. We will first consider $k=2$. The \emph{new} triangles in the mod $p^2$ shape can be indexed as follows:
\begin{align*}
    \binom{l\cdot p^n + s\cdot p^{n-1} + r}{q\cdot p^n + t \cdot p^{n-1} + c} & \text{ with }
    \begin{array}{l}
        0 \leq q < l < p \; ,\\
        0 \leq s < t < p \; ,\\
        0 \leq c \leq r < p^{n-1}.
    \end{array} 
\end{align*}
These lie within the \emph{empty} regions of the mod $p$ triangle. Similar to the proof of Claim \ref{claim:valuationdigits} we can let the carries $c^{r',c'}_{i}$ be defined as before, for the numbers $r'=[l\;s\;r_{n-2}r_{n-3}...r_0]_p$ and $c'=[q\;t\;c_{n-2}...c_0]_p$. We have $q<l$ hence $c^{r',c'}_{n}=0$ and $s<t$ so $c^{r',c'}_{n-1}=1$. This means there is exactly one extra carry compared to $r$ and $c$ and hence by Kummer's theorem
\begin{align*}
    \nu_p(\binom{l\cdot p^n + s\cdot p^{n-1} + r}{q\cdot p^n + t \cdot p^{n-1} + c}) = \nu_p(\binom{r}{c}) + 1
\end{align*}
for these values of $l,q,s,t,r,c$.
This is what we need, because it implies
\begin{align*}
    \binom{l\cdot p^n + s\cdot p^{n-1} + r}{q\cdot p^n + t \cdot p^{n-1} + c} \equiv 0 \mod p^{k+1} \iff \binom{r}{c} \equiv 0 \mod p^{k}
\end{align*}
i.e. in the mod $p^{k+1}$ shape there is a copy of the (smaller) mod $p^k$ shape. This proves what was drawn as a size-$p^n$ triangle in Figure \ref{fig:valuations}.
If we continue to the triangles that are newly added in the `mod $3^3$ triangle', we find that they correspond to
\begin{align*}
    &\binom{r'_n\cdot p^n + r'_{n-1}\cdot p^{n-1} + r'_{n-2}\cdot p^{n-2} + r}{c'_{n}\cdot p^n + c'_{n-1}\cdot p^{n-1} + c'_{n-2}\cdot p^{n-2} + c} \\
    &\qquad\qquad \text{ with }
    \begin{array}{l}
        0 \leq c'_n < r'_n < p \; , \\
        0 \leq r'_{n-1} \leq c'_{n-1} < p \; , \\
        0 \leq r'_{n-2} < c'_{n-2} < p \; , \\
        0 \leq c \leq r < p^{n-2}.
    \end{array} 
\end{align*}
Define $r'=[r'_n r'_{n-1} r'_{n-2} r_{n-3} \cdots r_0]_p$ and $c'=[c'_n c'_{n-1} c'_{n-2} c_{n-3} \cdots c_0]_p$. Then by the same reasoning as before we can apply Kummer's theorem to obtain $\nu_p(\binom{r'}{c'}) = \nu_p(\binom{r}{c})+ c^{r',c'}_{n-2} + c^{r',c'}_{n-1} + c^{r',c'}_{n}$. Looking at the constraints for digits $n-2$ up to $n$ we see that $c^{r',c'}_{n} = 0$, and $c^{r',c'}_{n-1}$ is $1$ or equal to $c^{r',c'}_{n-2}$ which is always $1$. We conclude: $\nu_3(\binom{r'}{c'}) = \nu_3(\binom{r}{c})+2$. We can continue the pattern, and we find that in the `mod $p^{k+1}$ triangle', the newly added triangles correspond to the following constraints on the digits of $r',c'$ with the following carries:
\begin{align*}
    0 \leq c'_n       &<    r'_n < p          & c^{r',c'}_{n} = 0 \\
    0 \leq r'_{n-1}   &\leq c'_{n-1} < p   & c^{r',c'}_{n-1} = 1 \text{ or } c^{r',c'}_{n-1} = c^{r',c'}_{n-2} \\
    &\vdots \\
    0 \leq r'_{n-k+1} &\leq c'_{n-k+1} < p & c^{r',c'}_{n-k+1} = 1 \text{ or } c^{r',c'}_{n-k+1} = c^{r',c'}_{n-k} \\
    0 \leq r'_{n-k}   &<    c'_{n-k} < p      & c^{r',c'}_{n-k} = 1 \\
    0 \leq c          &\leq r < p^{n-k}   & \nu_3\binom{r}{c} .
\end{align*}
We see that $\nu_3\binom{r'}{c'} = \nu_3\binom{r}{c} + k$ as required.
We still have to show that the empty regions in the `mod $p^{k+1}$ triangle' are empty. They correspond to the same indices as above except for $0\leq r'_{n-k} \leq c'_{n-k} < p$ and $0 \leq r < c < p^{n-k}$. We can apply the same idea as in the proof of Claim \ref{claim:valuationdigits} by noting that the first digit where $r$ and $c$ differ will satisfy $r_{i^*} < c_{i^*}$ and hence all the carries $c^{r',c'}_i$ are $1$ for $i \geq i^*$. This gives $\nu_3 \binom{r'}{c'} \geq k+1$, meaning that $\binom{r'}{c'}\equiv 0 \mod p^{k+1}$ so the region is indeed empty.

~

Since the numbers in Pascal's triangle can be thought of as scaled probabilities of a random walk, one could imagine writing down probabilities of a quantum walk, scaled to become integer, and show them modulo two. The next section will introduce a specific quantum walk and apply this idea with $p=2$ and $p=3$.

\section{Hadamard Walk}
Quantum walks are simple models for a quantum particle moving through some system. This paper is only concerned with the probability distribution that emerges from one particular quantum walk, and therefore the physical aspects of it are left out. Here we only provide a short overview of the relevant concepts, and we refer the reader to \cite{Nielsen} for a complete introduction to the field of quantum information. For the purposes of this paper we only need to know that the state of a particle is described by a unit vector in a complex Hilbert space, and quantum mechanics dictates that time evolution is limited to applying unitary operators to this vector. We will denote such state vectors using the commonly used `bra-ket' notation, writing $\ket{\psi}$ for a vector as opposed to $\vec{\psi}$. We can write the vector as a linear combination of orthonormal basis states, $\ket{\psi}=\sum_i \alpha_i \ket{x_i}$ where $\alpha_i\in\mathbb{C}$ and $\sum_i |\alpha_i|^2 = 1$. The $\ket{x_i}$ are the standard basis vectors of the Hilbert space and the coefficients $\alpha_i$ are known as \emph{amplitudes}. Inner products are denoted by $\langle\phi\vert\psi\rangle$ for two vectors $\ket{\phi}$ and $\ket{\psi}$. One of the axioms of quantum mechanics states that one can observe (measure) the system $\ket{\psi}$ in a chosen basis and the result can be any of the basis states, where state $\ket{x_i}$ has probability $|\alpha_i|^2$ of appearing.

A simple example of a quantum walk on a one-dimensional line is the so-called Hadamard walk \cite{Ambainis01}. It can be thought of as a quantum particle moving on $\mathbb{Z}$, the discrete line. The particle has an internal degree of freedom other than its position (a spin-$\frac{1}{2}$ degree of freedom, for physicists). The internal state is sometimes referred to as the \emph{coin state} of the particle with associated Hilbert space $\mathcal{H}_\mathrm{coin}=\mathbb{C}^2$ and basis states $\spinup$ and $\spindown$. The Hilbert space associated with the complete quantum system is $\mathcal{H}=\mathcal{H}_\mathrm{pos}\otimes\mathcal{H}_\mathrm{coin}$ where $\mathcal{H}_\mathrm{pos}=\mathrm{span}\{ \ket{n} \;|\; n\in\mathbb{Z} \}$. So the most general state of the particle is
\[
    \ket{\psi} = \sum_{n\in\mathbb{Z}} \left(\alpha_{n,\uparrow}\ket{n,\uparrow}+\alpha_{n,\downarrow}\ket{n,\downarrow}\right) ,
\]
with normalization $\sum_{n\in\mathbb{Z}}\left(|\alpha_{n,\uparrow}|^2+|\alpha_{n,\downarrow}|^2\right) = 1$ and where we used the notation $\ket{a,b}\equiv \ket{a}\otimes\ket{b}$. The dynamics of the particle are given by repeated application of a unitary operator $U$ that consists of two steps. The first step is a unitary only applied to the internal state and is sometimes considered the quantum analogue of `flipping a coin'. The second step updates the position of the particle conditioned on the outcome of the coin. In the specific case of the Hadamard walk, the unitary in the first step is the Hadamard operator $H$, defined as
\begin{align*}
    H = \frac{1}{\sqrt{2}}\begin{pmatrix} 1 & 1 \\ 1 & -1 \end{pmatrix}, \text{ where } \spinup = \begin{pmatrix}1\\ 0\end{pmatrix} \text{ and } \spindown = \begin{pmatrix}0\\ 1\end{pmatrix}.
\end{align*}
The time evolution operator $U$ is then given by
\begin{align*}
    U = S\cdot(\mathrm{Id}_\mathrm{pos} \otimes H),
\end{align*}
where $\mathrm{Id}_\mathrm{pos}$ is the identity on $\mathcal{H}_\mathrm{pos}$ and $S$ is called the \emph{shift}, given by
\begin{align*}
    S\ket{n,\uparrow} = \ket{n+1,\uparrow},\qquad S\ket{n,\downarrow}=\ket{n-1,\downarrow}.
\end{align*}
It can be thought of as updating the position of the particle conditioned on the outcome of the coin flip. A full step $U$ of the Hadamard walk, acting on the basis state $\ket{n,\downarrow}$ for example, is given by
\[
U\ket{n,\downarrow} \overset{\mathrm{coin}}{=} S\;\left(\frac{1}{\sqrt{2}}\ket{n,\uparrow} - \frac{1}{\sqrt{2}}\ket{n,\downarrow}\right) \overset{\mathrm{shift}}{=} \frac{1}{\sqrt{2}} \ket{{n+1},\uparrow} - \frac{1}{\sqrt{2}} \ket{n-1,\downarrow} .
\]
If we now were to measure the system, the result would be either $\ket{n+1,\uparrow}$ or $\ket{n-1,\downarrow}$, both with probability $|{\pm1}/\sqrt{2}|^2=1/2$.
\begin{figure}
    \begin{center}
        \begin{tikzpicture}
            \foreach \x in {0}
            {
                \draw[shorten >=0.1cm,->] (\x, 1) -- node[above] {$\frac{ 1}{\sqrt{2}}$} (\x+1, 1);
                \draw[shorten >=0.1cm,->] (\x, 1) -- node[left]  {$\frac{ 1}{\sqrt{2}}$} (\x-1, 0);
                \draw[shorten >=0.1cm,->] (\x, 0) -- node[right] {$\frac{ 1}{\sqrt{2}}$} (\x+1, 1);
                \draw[shorten >=0.1cm,->] (\x, 0) -- node[below] {$\frac{-1}{\sqrt{2}}$} (\x-1, 0);
            }
            \foreach \x in {-3,...,3}
            {
                \draw[fill] (\x, 0) circle (0.03);
                \draw[fill] (\x, 1) circle (0.03);
                \node at (\x,2.3) {$|{\x}\rangle_\mathrm{p}$};
            }
            \node at (4.0,1) {$\spinup_\mathrm{c}$};
            \node at (4.0,0) {$\spindown_\mathrm{c}$};
        \end{tikzpicture}
    \end{center}
    \caption{\label{fig:diagram1}Graphical representation of one step $U$ of the Hadamard walk. The dots represent possible quantum states of the form $\ket{n,\uparrow}$ (top row) and $\ket{n,\downarrow}$ (bottom row) and the arrows represent one application of the coin and shift when starting at position 0. The arrow going from $\ket{0,\downarrow}$ to $\ket{{-1},\downarrow}$ represents that the amplitude at $\ket{0,\downarrow}$ is multiplied by ${-1}/\sqrt{2}$ and then stored at $\ket{{-1},\downarrow}$, added to a part of the amplitude coming from $\ket{0,\uparrow}$. }
\end{figure}
Figure \ref{fig:diagram1} shows a schematic representation of one step $U$ of the Hadamard walk.

With these definitions, one can now consider the following process. Select a starting state, say $\ket{\psi_s}=\ket{0,\uparrow}$, evolve it with $U$ for $t$ steps and then measure the position. For example, starting in $\ket{0,\uparrow}$, the state of the system after three steps is given by
\begin{align*}
    U^3 \ket{0,\uparrow} = \frac{1}{(\sqrt{2})^3}\Big(
      \ket{{-3},\downarrow}
    - \ket{{-1},\uparrow}
    + 2 \ket{1,\uparrow} + \ket{1,\downarrow}
    + \ket{3,\uparrow}
    \Big).
\end{align*}
Now measuring the system will result in finding some position $X$ with probabilities $\mathbb{P}[X{=}{-3}]=\frac{1}{8}$, $\mathbb{P}[X{=}{-1}]=\frac{1}{8}$, $\mathbb{P}[X{=}1]=\frac{5}{8}$ and $\mathbb{P}[X{=}3]=\frac{1}{8}$. The amplitudes of the first five steps are also displayed in Figure \ref{fig:qpascal2}.

\subsection{Hadamard triangle}
The numbers in Pascal's triangle can be thought of as scaled probabilities of a random walk, and carrying this idea over to the Hadamard walk, one could consider the amplitudes or probabilities of the Hadamard walk, but scaled by a factor of $\sqrt{2^n}$ so that all numbers involved become integer. Another way to view this is instead of applying $H$, use $\sqrt{2}H$, a matrix with only integer coefficients. Note that we could either use the \emph{amplitudes} or the \emph{probabilities} which are simply their squares. However, since we are primarily interested in whether or not they are divisible by some prime $p$, squaring the amplitudes does not make a difference. We therefore continue with the (unsquared) amplitudes. Figure \ref{fig:qpascal2} shows the start of the Hadamard triangle.
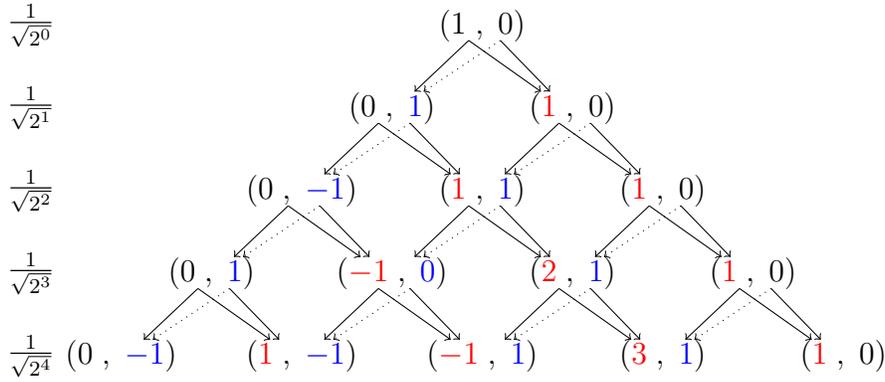
\begin{figure}
    \begin{center}
        \begin{tikzpicture}
            \def\xscale{1.2};
            \def\scale{1.1};
            
            \node at (-5*\xscale,  0*\scale) {$\frac{1}{\sqrt{2^0}}$};
            \node at (-5*\xscale, -1*\scale) {$\frac{1}{\sqrt{2^1}}$};
            \node at (-5*\xscale, -2*\scale) {$\frac{1}{\sqrt{2^2}}$};
            \node at (-5*\xscale, -3*\scale) {$\frac{1}{\sqrt{2^3}}$};
            \node at (-5*\xscale, -4*\scale) {$\frac{1}{\sqrt{2^4}}$};
            
            \node at ( 0*\xscale,  0*\scale) {$\twovect{ 1}{ 0}$};
            \node at (-1*\xscale, -1*\scale) {$\twovect{ 0}{ \color{blue}{1}}$};
            \node at (+1*\xscale, -1*\scale) {$\twovect{ \color{red}{1}}{ 0}$};
            
            \node at (-2*\xscale, -2*\scale) {$\twovect{             0}{\color{blue}{-1}}$};
            \node at ( 0*\xscale, -2*\scale) {$\twovect{\color{red}{1}}{\color{blue}{ 1}}$};
            \node at (+2*\xscale, -2*\scale) {$\twovect{\color{red}{1}}{ 0}$};
            
            \node at (-3*\xscale, -3*\scale) {$\twovect{              0}{\color{blue}{1}}$};
            \node at (-1*\xscale, -3*\scale) {$\twovect{\color{red}{-1}}{\color{blue}{0}}$};
            \node at (+1*\xscale, -3*\scale) {$\twovect{\color{red}{ 2}}{\color{blue}{1}}$};
            \node at (+3*\xscale, -3*\scale) {$\twovect{\color{red}{ 1}}{ 0}$};
            
            \node at (-4*\xscale, -4*\scale) {$\twovect{              0}{\color{blue}{-1}}$};
            \node at (-2*\xscale, -4*\scale) {$\twovect{\color{red}{ 1}}{\color{blue}{-1}}$};
            \node at ( 0*\xscale, -4*\scale) {$\twovect{\color{red}{-1}}{\color{blue}{ 1}}$};
            \node at (+2*\xscale, -4*\scale) {$\twovect{\color{red}{ 3}}{\color{blue}{ 1}}$};
            \node at (+4*\xscale, -4*\scale) {$\twovect{\color{red}{ 1}}{               0}$};           

            \foreach \y in {0}
            {
            \foreach \x in {0}
            {
                \draw[->] (\x*\xscale-0.15*\xscale,-\y*\scale-0.2) -- (\x*\xscale - 0.75*\xscale,-\y*\scale-0.8*\scale);
                \draw[dotted,->] (\x*\xscale+0.20*\xscale,-\y*\scale-0.2) -- (\x*\xscale - 0.65*\xscale,-\y*\scale-0.8*\scale);
                \draw[->] (\x*\xscale-0.15*\xscale,-\y*\scale-0.2) -- (\x*\xscale + 0.65*\xscale,-\y*\scale-0.8*\scale);
                \draw[->] (\x*\xscale+0.20*\xscale,-\y*\scale-0.2) -- (\x*\xscale + 0.75*\xscale,-\y*\scale-0.8*\scale);
            }
            }
            \foreach \y in {1}
            {
            \foreach \x in {-1, 1}
            {
                \draw[->] (\x*\xscale-0.15*\xscale,-\y*\scale-0.2) -- (\x*\xscale - 0.75*\xscale,-\y*\scale-0.8*\scale);
                \draw[dotted,->] (\x*\xscale+0.20*\xscale,-\y*\scale-0.2) -- (\x*\xscale - 0.65*\xscale,-\y*\scale-0.8*\scale);
                \draw[->] (\x*\xscale-0.15*\xscale,-\y*\scale-0.2) -- (\x*\xscale + 0.65*\xscale,-\y*\scale-0.8*\scale);
                \draw[->] (\x*\xscale+0.20*\xscale,-\y*\scale-0.2) -- (\x*\xscale + 0.75*\xscale,-\y*\scale-0.8*\scale);
            }
            }
            \foreach \y in {2}
            {
            \foreach \x in {-2, 0, 2}
            {
                \draw[->] (\x*\xscale-0.15*\xscale,-\y*\scale-0.2) -- (\x*\xscale - 0.75*\xscale,-\y*\scale-0.8*\scale);
                \draw[dotted,->] (\x*\xscale+0.20*\xscale,-\y*\scale-0.2) -- (\x*\xscale - 0.65*\xscale,-\y*\scale-0.8*\scale);
                \draw[->] (\x*\xscale-0.15*\xscale,-\y*\scale-0.2) -- (\x*\xscale + 0.65*\xscale,-\y*\scale-0.8*\scale);
                \draw[->] (\x*\xscale+0.20*\xscale,-\y*\scale-0.2) -- (\x*\xscale + 0.75*\xscale,-\y*\scale-0.8*\scale);
            }
            }
            \foreach \y in {3}
            {
            \foreach \x in {-3, -1, 1, 3}
            {
                \draw[->] (\x*\xscale-0.15*\xscale,-\y*\scale-0.2) -- (\x*\xscale - 0.75*\xscale,-\y*\scale-0.8*\scale);
                \draw[dotted,->] (\x*\xscale+0.20*\xscale,-\y*\scale-0.2) -- (\x*\xscale - 0.65*\xscale,-\y*\scale-0.8*\scale);
                \draw[->] (\x*\xscale-0.15*\xscale,-\y*\scale-0.2) -- (\x*\xscale + 0.65*\xscale,-\y*\scale-0.8*\scale);
                \draw[->] (\x*\xscale+0.20*\xscale,-\y*\scale-0.2) -- (\x*\xscale + 0.75*\xscale,-\y*\scale-0.8*\scale);
            }
            }
        \end{tikzpicture}
    \end{center}
    \caption{\label{fig:qpascal2}The up- and down-components of the amplitudes of the first 5 steps of the Hadamard walk, starting in $\ket{0,\uparrow}$. Every row corresponds to one time-step. The normalisation of each row is shown at the left side. At even timesteps, only the even positions are shown and at odd time-steps only the odd positions are shown, similar to Figure \ref{fig:pascal1}. The arrows represent the time-step of the Hadamard walk. A dotted arrow means the incoming amplitude is multiplied by ${-1}$ before being added to the other incoming amplitude. The red and blue colouring denotes a subset of amplitudes that is used in Section \ref{sec:hadamardmod2} and \ref{sec:hadamardmod3}.}
\end{figure}
We will now derive expressions for these amplitudes when starting in $\ket{0,\uparrow}$.

\subsection{Expressions for amplitudes}
Meyer \cite{Meyer96} gave explicit expressions for the amplitudes encountered in the Hadamard walk. Let $\psi_\uparrow(n,t)$ be the amplitude at $\ket{n,\uparrow}$ after $t$ steps when starting in $\ket{0,\uparrow}$, i.e. $\psi_\uparrow(n,t):=\langle n,\uparrow\!|U^t|0,\uparrow\rangle$. Similarly, let $\psi_\downarrow(n,t):=\langle n,\downarrow\!|U^t|0,\uparrow\rangle$. Then we have
\begin{lemma}[\cite{Meyer96}] When $t+n$ is odd or when $|n|>t$ we have $\psi_\uparrow(n,t)=\psi_\downarrow(n,t)=0$. Otherwise the amplitudes are given by
\begin{align*}
    \psi_\uparrow  (n,t) &=
    \begin{cases}
    \frac{1}{\sqrt{2^t}} & n = t\\
    \frac{1}{\sqrt{2^t}} \sum_{k\geq 1} \binom{(t-n)/2-1}{k-1}\binom{(t+n)/2}{k}(-1)^{(t-n)/2-k} & n < t
    \end{cases}\\
    \psi_\downarrow(n,t) &= \frac{1}{\sqrt{2^t}} \sum_{k\geq 0} \binom{(t-n)/2-1}{k}\binom{(t+n)/2}{k}(-1)^{(t-n)/2-k-1}
\end{align*}
\end{lemma}
We will give an alternative and slightly shorter proof of this for a general coin operator. This proof also allows us to make another observation stated in the following claim. Let $C$ be any unitary 2x2 matrix. Any such matrix can be written as follows
\begin{align*}
    C = \begin{pmatrix}
        c_r & c_u\\
        c_d & c_l
    \end{pmatrix}
    = \begin{pmatrix}
        \sqrt{p} \; e^{i \alpha} & \sqrt{1-p} \; e^{i\beta} \\
        - \sqrt{1-p} \; e^{i \gamma} & \sqrt{p} \; e^{i (\gamma + \beta - \alpha)}
    \end{pmatrix} \quad , \quad \text{with } 0 \leq p \leq 1.
\end{align*}
\begin{claim} \label{claim:amplitudeexpressions}
    Let $\psi_\uparrow(n,t)$ and $\psi_\downarrow(n,t)$ be the up and down amplitudes at position $n$ at time $t$ but for the general coin operator $C$.
    When $t+n$ is odd or when $|n|>t$ we have $\psi_\uparrow(n,t)=\psi_\downarrow(n,t)=0$. Otherwise the amplitudes are given by
    \begin{align*}
        \psi_\uparrow  (n,t) &=
        \begin{cases}
            e^{i\alpha n} \sqrt{p^t} & n = t\\
            e^{i\left(\alpha n + (\gamma+\beta)(t-n)/2\right)} \sqrt{p^t} \sum_{k\geq 1} \binom{(t+n)/2}{k}\binom{(t-n)/2-1}{k-1} \left(-\frac{1-p}{p}\right)^k & n < t
        \end{cases}\\
        \psi_\downarrow(n,t) &= -e^{i(\alpha n +(\gamma+\beta)(t-n)/2 -\beta)} \\
        &\qquad \qquad \times \quad \sqrt{(1-p)p^{t-1}} \sum_{k\geq 0} \binom{(t+n)/2}{k} \binom{(t-n)/2-1}{k} \left(-\frac{1-p}{p}\right)^k
    \end{align*}
    The probabilities $|\psi_\uparrow(n,t)|^2$ and $|\psi_\downarrow(n,t)|^2$ associated to these amplitudes are independent of the complex phases $\alpha,\beta,\gamma$ of the coin operator.
\end{claim}
Note that the lemma follows directly from the claim by setting $p=1/2$, $\alpha=\beta=0$ and $\gamma=\pi$ to obtain the Hadamard coin matrix. Furthermore note that for a more general starting state, not equal to $\ket{0,\uparrow}$, the probabilities \emph{do} depend on the complex phases present in the coin operator.

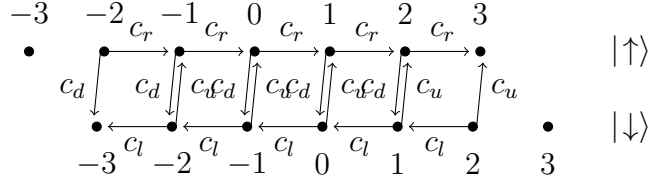
\begin{figure}
    \begin{center}
        \begin{tikzpicture}
            \def\xshift{0.9};
            \foreach \x in {-2,-1,0,1,2}
            {
                \draw[shorten >=0.15cm,->] (\x, 1)              -- node[above] {$c_r$} (\x+1, 1);
                \draw[shorten >=0.15cm,->] (\x-0.05, 1)         -- node[left]  {$c_d$} (\x+\xshift-1.05, 0);
                \draw[shorten >=0.15cm,->] (\x+\xshift+0.05, 0) -- node[right] {$c_u$} (\x+1.05, 1);
                \draw[shorten >=0.15cm,->] (\x+\xshift, 0)      -- node[below] {$c_l$} (\x+\xshift-1, 0);
            }
            \foreach \x in {-3,...,3}
            {
                \draw[fill] (\x+\xshift, 0) circle (0.06);
                \draw[fill] (\x, 1) circle (0.06);
                \node at (\x,1.5) {$\x$};
                \node at (\x+\xshift,-0.5) {$\x$};
            }
            \node at (5.0,1) {$\spinup$};
            \node at (5.0,0) {$\spindown$};
        \end{tikzpicture}
    \end{center}
    \caption{\label{fig:diagram2}Schematic representation of one step $U$ of the walk with generic coin. It is similar to Figure \ref{fig:diagram1} but the image is slightly tilted and for a general coin.}
\end{figure}
We want to have an expression for both the up and down component at position $n$ after $t$ steps, i.e. we want to know $\langle n,\uparrow\!|U^t|0,\uparrow\rangle$ and $\langle n,\downarrow\!|U^t|0,\uparrow\rangle$. These two cases will be handled separately. We will use the path counting technique, where one counts all possible paths starting at $\ket{0,\uparrow}$ and end at the desired state. Every path gets a certain (possibly negative) amplitude and these have to be added.
\subsubsection{Up component}
We want to find all possible paths from $\ket{0,\uparrow}$ to $\ket{n,\uparrow}$ using $t$ steps. If $n=t$ there is exactly one path. Otherwise, we have $-t< n < t$. Figure \ref{fig:diagram2} shows the directed graph on which we consider possible paths. Consider a single path and let $r,l,u,d$ be the number of times the path uses the right, left, up and down arrows respectively. To end at $\ket{n,\uparrow}$ we then have
\begin{align*}
    r+l+u+d&=t &\text{total number of steps}\\
    r-l &=n    &\text{ending column}\\
    u   &=d    &\text{start up and end up}
\end{align*}
Let $k=u=d$, then we have $r=\frac{t+n}{2} - k$ and $l=\frac{t-n}{2} - k$. We have $k\geq 1$ (we need to go down and up at least once) and $k\leq\frac{t-n}{2},\frac{t+n}{2}$.\\

For a specific set of values $(r,l,u,d)$ the path will arrive with an amplitude $(c_r)^r(c_l)^l(c_u)^u(c_d)^d$. So we sum over all possible values for $(r,l,u,d)$ and count how many paths there are for a specific set of values $(r,l,u,d)$. We can construct such paths as follows. Construct a sequence of choices to make if the walker is in the top layer and another sequence of choices to make if the walker is in the bottom layer.
The walker is in the $\spinup$ state (top layer) $k+r=\frac{t+n}{2}$ times, out of which $r$ times it goes right and $k$ times it goes down. There are $\binom{(t+n)/2}{k}$ possible ways to do this. Likewise, the particle is in $\spindown$ (bottom layer) $l+k=\frac{t-n}{2}$ times and has to choose between left and up. The \emph{last} of these choices should always be up, so this gives $\binom{(t-n)/2-1}{k-1}$ possibilities. To construct the full path, start with the top-layer choices, and whenever the choice is `down', continue with the bottom-layer choices and so on.
Therefore
\begin{align*}
    \bra{n,\uparrow}U^t\ket{0,\uparrow} &= \begin{cases} (c_r)^t & n = t\\
    \sum_{k\geq 1} \binom{(t+n)/2}{k}\binom{(t-n)/2-1}{k-1} c_r^{(t+n)/2-k} c_l^{(t-n)/2-k} c_u^k c_d^k & n < t \end{cases}.
\end{align*}
Now rewrite the last sum for $n<t$ and group the factors that do not depend on $k$:
\begin{align*}
    c_r^{(t+n)/2} c_l^{(t-n)/2} \sum_{k\geq 1} \binom{(t+n)/2}{k}\binom{(t-n)/2-1}{k-1} \left(\frac{c_u c_d}{c_r c_l}\right)^k
\end{align*}
Note that $\frac{c_u c_d}{c_r c_l} = -\frac{1-p}{p}$ so this fraction is always a real (negative) number, regardless of the complex phases present in the entries of the coin matrix. The sum above in terms of $p$ and $\alpha,\beta,\gamma$ is equal to
\begin{align*}
    e^{i\left(\alpha n + (\gamma+\beta)(t-n)/2\right)} \sqrt{p^t} \sum_{k\geq 1} \binom{(t+n)/2}{k}\binom{(t-n)/2-1}{k-1} \left(-\frac{1-p}{p}\right)^k ,
\end{align*}
as claimed.
The probability $|\psi_\uparrow(n,t)|^2$ of being at $\ket{n,\uparrow}$ after $t$ steps when starting in $\ket{0,\uparrow}$ is independent of $\alpha,\beta,\gamma$ since the only dependence on these variables is in the prefactor $e^{i\left(\alpha n + (\gamma+\beta)(t-n)/2\right)}$ which always has norm 1.

\subsubsection{Down component}
For the down component, the equations are similar:
\begin{align*}
    r+l+u+d&=t   &\text{total number of steps}\\
    r-l    &=n+1 &\text{ending column (tilted)}\\
    u+1    &=d   &\text{start up and end down}
\end{align*}
Let $u=k$, then the equations give $r=(t+n)/2-k$ and $l=(t-n)/2-k-1$. The argument is the same as before, but now the last choice in the top layer has to be `down' with no restrictions on the last choice in the bottom layer.
We are in $\spinup$ $r+d=(t+n)/2+1$ times. The last choice has to be down, so this gives $\binom{(t+n)/2}{k}$.
We are in $\spindown$ $l+u=(t-n)/2-1$ times which gives $\binom{(t-n)/2-1}{k}$. The expression is therefore given by
\begin{align*}
    \psi_\downarrow(n,t) &= \sum_{k\geq 0} \binom{(t+n)/2}{k} \binom{(t-n)/2-1}{k} c_u^k c_d^{k+1} c_l^{(t-n)/2-k-1} c_r^{(t+n)/2-k} .
\end{align*}
Rewriting this in terms of $p,\alpha,\beta,\gamma$ gives the expression given in the claim. Again the only dependency on the complex phases is in the prefactor which has norm 1 so the probability $|\psi_\downarrow(n,t)|^2$ only depends on $p$.

\subsection{Hadamard walk modulo 2 - Sierpinski triangle} \label{sec:hadamardmod2}
When the amplitudes of the Hadamard walk are plotted modulo two, the Sierpinski triangle appears in a similar fashion to Pascal's triangle. To see why this is the case, we note that to find the amplitudes at some time $t$ modulo two it is enough to consider a process where every single time-step is done modulo two. The scaled Hadamard operator becomes
\begin{align*}
    \sqrt{2}H \equiv \begin{pmatrix} 1 & 1 \\ 1 & 1 \end{pmatrix} \mod 2,
\end{align*}
and we can immediately see that the amplitude sent to the right is the same as the amplitude sent to the left. More precisely, after any time-step the amplitude at $\ket{n-1,\downarrow}$ is the same as the amplitude at $\ket{n+1,\uparrow}$ modulo two, for all $n$. 
\begin{figure}
    \begin{center}
        \begin{tikzpicture}
            \def\xscale{0.8};
            \def\scale{0.8};
            
            \draw[fill,red!15!white] (0,-1*\scale)          ellipse ({0.92*\xscale} and {0.5*\scale});
            \draw[fill,red!15!white] (-1*\xscale,-2*\scale) ellipse ({0.92*\xscale} and {0.5*\scale});
            \draw[fill,red!15!white] ( 1*\xscale,-2*\scale) ellipse ({0.92*\xscale} and {0.5*\scale});
            \draw[fill,red!15!white] (-2*\xscale,-3*\scale) ellipse ({0.92*\xscale} and {0.5*\scale});
            \draw[fill,red!15!white] ( 0*\xscale,-3*\scale) ellipse ({0.92*\xscale} and {0.5*\scale});            
            \draw[fill,red!15!white] ( 2*\xscale,-3*\scale) ellipse ({0.92*\xscale} and {0.5*\scale});
            \draw[fill,red!15!white] (-3*\xscale,-4*\scale) ellipse ({0.92*\xscale} and {0.5*\scale});
            \draw[fill,red!15!white] (-1*\xscale,-4*\scale) ellipse ({0.92*\xscale} and {0.5*\scale});
            \draw[fill,red!15!white] ( 1*\xscale,-4*\scale) ellipse ({0.92*\xscale} and {0.5*\scale});
            \draw[fill,red!15!white] ( 3*\xscale,-4*\scale) ellipse ({0.92*\xscale} and {0.5*\scale});
            
            \node at ( 0*\xscale,  0*\scale) {$\twovect{ 1}{ 0}$};
            \node at (-1*\xscale, -1*\scale) {$\twovect{ 0}{ 1}$};
            \node at (+1*\xscale, -1*\scale) {$\twovect{ 1}{ 0}$};
            
            \node at (-2*\xscale, -2*\scale) {$\twovect{ 0}{ 1}$};
            \node at ( 0*\xscale, -2*\scale) {$\twovect{ 1}{ 1}$};
            \node at (+2*\xscale, -2*\scale) {$\twovect{ 1}{ 0}$};
            
            \node at (-3*\xscale, -3*\scale) {$\twovect{ 0}{ 1}$};
            \node at (-1*\xscale, -3*\scale) {$\twovect{ 1}{ 0}$};
            \node at (+1*\xscale, -3*\scale) {$\twovect{ 0}{ 1}$};
            \node at (+3*\xscale, -3*\scale) {$\twovect{ 1}{ 0}$};
            
            \node at (-4*\xscale, -4*\scale) {$\twovect{ 0}{ 1}$};
            \node at (-2*\xscale, -4*\scale) {$\twovect{ 1}{ 1}$};
            \node at ( 0*\xscale, -4*\scale) {$\twovect{ 1}{ 1}$};
            \node at (+2*\xscale, -4*\scale) {$\twovect{ 1}{ 1}$};
            \node at (+4*\xscale, -4*\scale) {$\twovect{ 1}{ 0}$};           
            
            \foreach \y in {0}
            {
            \foreach \x in {0}
            {
                \draw[->] (\x*\xscale-0.15*\xscale,-\y*\scale-0.2) -- (\x*\xscale - 0.75*\xscale,-\y*\scale-0.8*\scale);
                \draw[->] (\x*\xscale+0.15*\xscale,-\y*\scale-0.2) -- (\x*\xscale - 0.65*\xscale,-\y*\scale-0.8*\scale);
                \draw[->] (\x*\xscale-0.15*\xscale,-\y*\scale-0.2) -- (\x*\xscale + 0.65*\xscale,-\y*\scale-0.8*\scale);
                \draw[->] (\x*\xscale+0.15*\xscale,-\y*\scale-0.2) -- (\x*\xscale + 0.75*\xscale,-\y*\scale-0.8*\scale);
            }
            }
            \foreach \y in {1}
            {
            \foreach \x in {-1, 1}
            {
                \draw[->] (\x*\xscale-0.15*\xscale,-\y*\scale-0.2) -- (\x*\xscale - 0.75*\xscale,-\y*\scale-0.8*\scale);
                \draw[->] (\x*\xscale+0.15*\xscale,-\y*\scale-0.2) -- (\x*\xscale - 0.65*\xscale,-\y*\scale-0.8*\scale);
                \draw[->] (\x*\xscale-0.15*\xscale,-\y*\scale-0.2) -- (\x*\xscale + 0.65*\xscale,-\y*\scale-0.8*\scale);
                \draw[->] (\x*\xscale+0.15*\xscale,-\y*\scale-0.2) -- (\x*\xscale + 0.75*\xscale,-\y*\scale-0.8*\scale);
            }
            }
            \foreach \y in {2}
            {
            \foreach \x in {-2, 0, 2}
            {
                \draw[->] (\x*\xscale-0.15*\xscale,-\y*\scale-0.2) -- (\x*\xscale - 0.75*\xscale,-\y*\scale-0.8*\scale);
                \draw[->] (\x*\xscale+0.15*\xscale,-\y*\scale-0.2) -- (\x*\xscale - 0.65*\xscale,-\y*\scale-0.8*\scale);
                \draw[->] (\x*\xscale-0.15*\xscale,-\y*\scale-0.2) -- (\x*\xscale + 0.65*\xscale,-\y*\scale-0.8*\scale);
                \draw[->] (\x*\xscale+0.15*\xscale,-\y*\scale-0.2) -- (\x*\xscale + 0.75*\xscale,-\y*\scale-0.8*\scale);
            }
            }
            \foreach \y in {3}
            {
            \foreach \x in {-3, -1, 1, 3}
            {
                \draw[->] (\x*\xscale-0.15*\xscale,-\y*\scale-0.2) -- (\x*\xscale - 0.75*\xscale,-\y*\scale-0.8*\scale);
                \draw[->] (\x*\xscale+0.15*\xscale,-\y*\scale-0.2) -- (\x*\xscale - 0.65*\xscale,-\y*\scale-0.8*\scale);
                \draw[->] (\x*\xscale-0.15*\xscale,-\y*\scale-0.2) -- (\x*\xscale + 0.65*\xscale,-\y*\scale-0.8*\scale);
                \draw[->] (\x*\xscale+0.15*\xscale,-\y*\scale-0.2) -- (\x*\xscale + 0.75*\xscale,-\y*\scale-0.8*\scale);
            }
            }
        \end{tikzpicture}
    \end{center}
    \caption{\label{fig:qpascal3}Amplitudes of the first 5 steps of the scaled Hadamard walk modulo two, starting in $\ket{0,\uparrow}$. It is similar to Figure \ref{fig:qpascal2} but the amplitudes are considered modulo two. Every pair $(\cdot,\cdot)$ corresponds to up and down components of the state. The dotted arrow in Figure \ref{fig:qpascal2} becomes a normal arrow because ${-1}\equiv 1 \mod 2$. The red ellipses indicate pairs of values that are the same and form Pascal's triangle modulo two.}
\end{figure}
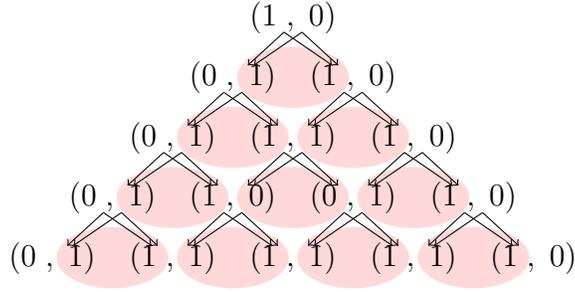
This idea is shown in Figure \ref{fig:qpascal3} which is similar to Figure \ref{fig:qpascal2} but modulo two. An ellipse is drawn around the pairs of amplitudes of states $\ket{n-1,\downarrow}$ and $\ket{n+1,\uparrow}$. The figure shows that the two values in each ellipse are equal, and are the sum of the values in the two neighbouring ellipses above it. This is the same rule with which Pascal's triangle can be constructed. Indeed, taking one value out of every ellipse, the Sierpinski triangle can be obtained. These are the either the red or the blue values shown in Figure \ref{fig:qpascal2}.

\subsection{Hadamard walk modulo 3 - Sierpinski carpet} \label{sec:hadamardmod3}
\begin{figure}
    \begin{center}
        \includegraphics[width=0.9\textwidth]{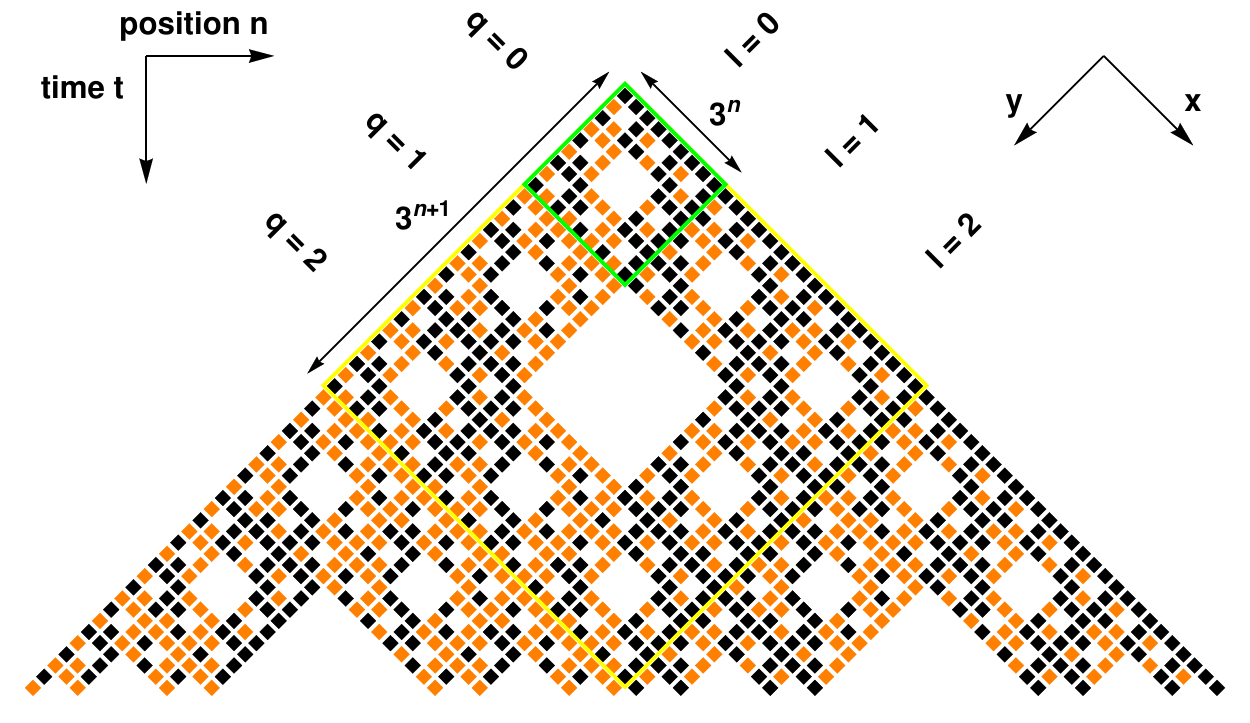}
    \end{center}
    \caption{\label{fig:carpet1}The start of the Sierpinski carpet resulting from colouring the scaled Hadamard walk amplitudes modulo 3. The horizontal direction is position and the vertical direction is time. The shape drawn at each point is a diamond, i.e. a rotated square instead of a square, because this gives a better visualisation of the $x,y$ coordinates.}
\end{figure}
We will now show that the $\spindown$ components of the scaled walk, modulo three, give rise to the Sierpinski carpet. In particular, we colour a square white if the amplitude is divisible by $3$, and give it a different colour otherwise. Figure \ref{fig:carpet1} shows the start of the resulting carpet. The top of the carpet is at $t=1$ and $n={-1}$ indicated by the \emph{blue} values in Figure \ref{fig:qpascal2}. Considering only these values, indexed by a row $R$ and column $C$ we are interested in amplitudes at $t=R+1$ and $n=2C-R-1$. Define $H_\mathrm{blue}(R,C):=\psi_\downarrow(2C-R-1,R+1)$, then
\begin{align}
    H_\mathrm{blue}(R,C) &= (-1)^{R-C} \sum_{k=0}^{\min(C,R-C)} \binom{C}{k}\binom{R-C}{k} (-1)^{k}. \label{eq:Hblue}
\end{align}
For the structure of the Sierpinski carpet, however, it is more convenient to consider coordinates $x,y$ that are aligned with the square structure of the carpet. The choice of these directions is indicated in Figure \ref{fig:carpet1}, and we have $R=x+y$ and $C=x$. As a function of these coordinates we define $\Phi(x,y) = H_\mathrm{blue}(x+y,x)$, so
\begin{align}
    \Phi(x,y) = (-1)^y \sum_{k=0}^{\min(x,y)} \binom{x}{k} \binom{y}{k} (-1)^{k}. \label{eq:phidef}
\end{align}
A pixel at coordinates $x,y$ is now coloured white if $\Phi(x,y)\equiv 0 \mod 3$ and a different colour otherwise. To show that the resulting figure is the Sierpinski carpet we have to show that for all $n\geq 1$ and $0\leq l,q \leq 2$ with $(l,q)\neq(1,1)$,
\begin{align}
    \Phi(x,y) \equiv \pm \Phi(l\cdot3^n + x,q\cdot3^n + y) \mod 3 \qquad \text{for all } 0\leq x,y \leq 3^n - 1 \label{eq:sim1}
\end{align}
where this means that for every $x,y,l,q$ the equivalence should hold with either a plus or minus sign. For $(l,q)=(1,1)$ we require
\begin{align}
    \Phi(3^n+x,3^n+y) \equiv 0 \mod 3 \qquad \text{for all } 0\leq x,y \leq 3^n - 1 \label{eq:sim2}
\end{align}
Figure \ref{fig:carpet1} shows this graphically. The values $(l,q)=(1,1)$ corresponds to the empty square in the middle, and all other values of $(l,q)$ should be copies of the square at $(l,q)=(0,0)$, up to exchanging the non-white colours.

To show this we prove something slightly more general.
\begin{definition} \label{def:fm}
For any $m\in\mathbb{Z}$ define $f_m:\mathbb{N}\times\mathbb{N}\to\mathbb{Z}$ as
\begin{align*}
    f_m(x,y) = \sum_{k=0}^{\min(x,y)} \binom{x}{k}\binom{y}{k} (-m)^k
\end{align*}
\end{definition}
The reason for the minus sign in $(-m)^k$ is that all valid quantum walks will have $m\geq 0$ this way, as will become clear later.
As a side note, this function is a special case of the so-called hypergeometric function $_2F_1(a,b;c;z)$, namely $f_m(x,y) = {_2}F_1(-x,-y,1,-m)$.
The following claim could be seen as something similar to Corollary \ref{cor:addpower} but for $f_m$:
\begin{claim} \label{claim:quantumlemma}
    Let $p$ be a prime and let $0\leq l,q \leq p-1$. Then we have for all $m\in\mathbb{Z}$ and for all $0\leq x,y\leq p^n - 1$
    \begin{align*}
        f_m(l\cdot p^n + x \;,\; q\cdot p^n + y) \equiv f_m(l,q)\cdot f_m(x,y) \mod p .
    \end{align*}
\end{claim}
\begin{proof}
    Note that any sum can be split in the following way:
    \begin{align*}
        \sum_{k=0}^{p^{n+1}-1} g(k) = \sum_{s=0}^{p-1} \sum_{k=0}^{p^n - 1} g(s\cdot p^n + k) ,
    \end{align*}
    where $s$ takes the role of the most significant digit and $k$ takes the role of the other digits. We apply this idea to the sum in $f_m(x,y)$ where we note that $\min(lp^n + x, qp^n + y) \leq p^{n+1} -1$ but we can let the sum range all the way to $p^{n+1}-1$ because the summand is zero in this extra range. Therefore we have
    \begin{align*}
        f_m(l\cdot p^n + x \;,\; q\cdot p^n + y)  &= \sum_{k=0}^{p^{n+1}-1} \binom{l\cdot p^n+x}{k}\binom{q\cdot p^n+y}{k} (-m)^k \\
                                                &= \sum_{s=0}^{p-1}\sum_{k=0}^{p^n-1} \binom{l\cdot p^n+x}{s\cdot p^n + k}\binom{q\cdot p^n+y}{s\cdot p^n + k} (-m)^{s\cdot p^n + k}
    \end{align*}
    Note that by Fermat's little theorem we have $m^p \equiv m \mod p$ so $m^{s\cdot p^n} \equiv m^s \mod p$. Now we apply Corollary \ref{cor:addpower} to the binomial coefficients to obtain
    \begin{align*}
        f_m(l\cdot p^n + x \;,\; q\cdot p^n + y)  &\equiv \sum_{s=0}^{p-1}\sum_{k=0}^{p^n-1} \binom{l}{s}\binom{q}{s} \binom{x}{k}\binom{y}{k} (-m)^{s + k} \\
                                                &\equiv \left( \sum_{s=0}^{p-1} \binom{l}{s}\binom{q}{s} (-m)^{s} \right ) f_m(x,y) \\
                                                &\equiv f_m(l,q)\cdot f_m(x,y) \mod p,
    \end{align*}
    as required.
\end{proof}
Note that $l,q$ take the role of the most significant digits and $x,y$ are the other digits. Just as Corollary \ref{cor:addpower} implies Lucas's theorem, we can apply this claim inductively on the number of digits to arrive at a result very similar to Lucas's theorem but now for the function $f_m$:
\begin{lemma2}[Lucas'-like theorem for $f_m$] \label{lemma:quantumlucas}
    Let $p$ be prime and $x,y$ non-negative integers. Let $x=[x_n x_{n-1}\cdots x_0]_p$ and $y=[y_n y_{n-1}\cdots y_0]_p$. Then for all $m\in\mathbb{Z}$ we have
    \begin{align*}
        f_m(x,y) \equiv f_m(x_n,y_n)\; f_m(x_{n-1},y_{n-1})\cdots f_m(x_0,y_0) \mod p.
    \end{align*}
\end{lemma2}

We can now prove \eqref{eq:sim1} and \eqref{eq:sim2} by noting that $\Phi(x,y) = (-1)^y f_{1}(x,y)$ so by Claim \ref{claim:quantumlemma} we have
\begin{align*}
    \Phi(l\cdot 3^n + x, q\cdot 3^n + y) \equiv (-1)^{q\cdot 3^n} f_{1}(l,q) \Phi(x,y) \equiv \Phi(l,q)\Phi(x,y) \mod 3.
\end{align*}
where we used that $(-1)^{q\cdot 3^n} = (-1)^q$. Note that $\Phi(1,1) = 0$ which proves \eqref{eq:sim2} and $\Phi(l,q) \equiv \pm 1 \mod 3$ for the other values of $l,q$ which proves \eqref{eq:sim1}.

\subsection{Results for a more general quantum walk}
We can generalize the results of the previous section. First of all, we can consider the same numbers modulo any prime $p$. But more generally, the Hadamard operator $H$ could be replaced by any matrix $C\in U(2)$. As stated before, we can write any unitary $2\times 2$ matrix as
\begin{align*}
    C = \begin{pmatrix}
        c_r & c_u\\
        c_d & c_l
    \end{pmatrix}
    = \begin{pmatrix}
        \sqrt{p} \; e^{i \alpha} & \sqrt{1-p} \; e^{i\beta} \\
        -\sqrt{1-p} \; e^{i \gamma} & \sqrt{p} \; e^{i (\gamma + \beta - \alpha)}
    \end{pmatrix} \quad , \quad \text{with } 0 \leq p \leq 1.
\end{align*}
The expression for the amplitudes with general coin operator is given by $\psi_\downarrow(n,t)$ as given in Claim \ref{claim:amplitudeexpressions}, and we can do the same substitutions as before to go to the $x,y$ coordinates:
\begin{align*}
    \Phi_C(x,y) &= c_d c_r^x c_l^y \sum_{k\geq 0} \binom{x}{k} \binom{y}{k} \left(-\frac{1-p}{p}\right)^k \\
               &= c_d c_r^x c_l^y \; f_{m}(x,y) .
\end{align*}
where $m=(1-p)/p$ and we extend the definition of $f_m$ for non-integer $m$. Note that $(1-p)/p\geq 0$ for any valid coin which was the reason for defining $f_m$ with a minus sign.
In the previous sections we considered the \emph{amplitudes} of the quantum walk as opposed to the \emph{probabilities} (which are equal to the norm squared of the amplitudes). For the purpose of the Sierpinski carpet, this distinction was not important because all entries of the Hadamard matrix are real and we were only interested in whether or not an integer was zero modulo a prime. Since $x\equiv 0 \mod p \iff x^2\equiv 0 \mod p$, squaring did not matter. For a general coin, however, there could be complex amplitudes and so we consider the corresponding \emph{probabilities} to make sure all numbers involved are real. Note that since $f_m$ is real, the imaginary component of $\Phi_C(x,y)$ comes only from $c_d c_r^x c_l^y$. We therefore consider the probabilities:
\begin{align*}
    |\Phi_C(x,y)|^2 = |c_d c_r^x c_l^y |^2 \left( f_{m}(x,y) \right)^2 .
\end{align*}
In the previous sections we rescaled the Hadamard matrix by a factor of $\sqrt{2}$ so that all numbers involved became integer. For a general coin matrix, in order to consider the probabilities modulo a prime, we assume that the coin matrix is such that $m=(1-p)/p$ is integer. Note that this can not be achieved by scaling the entire matrix because $m$ is invariant under such scalings. In fact, we have $p=\frac{1}{1+m}$ and $m\geq 0$ has to be integer. Furthermore, as stated in Claim \ref{claim:quantumlemma}, the complex phases $\alpha,\beta,\gamma$ do not influence $|\Phi_C(x,y)|^2$. Therefore the most general form of the matrix we can consider to obtain integer probabilities is the unitary matrix
\begin{align*}
    C_m =
    \begin{pmatrix}
        \sqrt{1/(1+m)}  & \sqrt{m/(1+m)} \\
        \sqrt{m/(1+m)}  & -\sqrt{1/(1+m)}
    \end{pmatrix}
    \quad \text{for } m\in\mathbb{Z} ,\; m \geq 0 ,
\end{align*}
where have set $\alpha=\beta=0$ and $\gamma=\pi$ such that $C_1=H$, but any other setting of phases would be equally valid.
If we want to scale the matrix by a factor $\lambda$ such that $|c_d c_r^x c_l^y|^2$ is integer then this requires $\lambda = \sqrt{n(1+m)}$ for any integer $n\geq 1$. This gives a scaled matrix 
\begin{align}
    \sqrt{n(1+m)}C_m = \sqrt{n}
    \begin{pmatrix}
        1 & \sqrt{m} \\
        \sqrt{m} & -1
    \end{pmatrix} , \label{eq:scaledgeneralcoin}
\end{align}
and for this scaled matrix, $|c_d c_r^x c_l^y|^2 = m n^{x+y+1}$.
By Claim \ref{claim:quantumlemma} we have for this scaled coin matrix that
\begin{align*}
    |\Phi_C(l\cdot p^n+x,q\cdot p^n+y)|^2 &\equiv \frac{n^{(l+q)(p^n-1)-1}}{m} \; |\Phi_C(l,q)|^2 \; |\Phi_C(x,y)|^2 \mod p \\
                                          &\equiv \frac{1}{mn} \; |\Phi_C(l,q)|^2 \; |\Phi_C(x,y)|^2 \mod p ,
\end{align*}
where we used Fermat's little theorem in the second step.
For $m=1$ and $n=1$ we recover the exact same rules as for the Hadamard matrix. This class also includes the commonly used coin
\begin{align*}
    \frac{1}{\sqrt{2}}\begin{pmatrix}1 & i \\ i & 1\end{pmatrix}.
\end{align*}
To find out what kind of fractals are generated by these quantum walks, it is useful to note that we are only interested in distinguishing $|\Phi(x,y)|^2 \equiv 0 \mod p$ from $|\Phi(x,y)|^2\not\equiv 0 \mod p$. Since we have $|\Phi(x,y)|^2 = mn^{x+y+1}(f_m(x,y))^2$ we can see that if $m\equiv 0 \mod p$ or $n\equiv 0 \mod p$ then all values $|\Phi(x,y)|^2$ are zero modulo $p$ and there is no fractal since all pixels are white. Therefore, assume that both $m$ and $n$ are not zero modulo $p$. In that case we have $|\Phi(x,y)|^2 \equiv 0\mod p$ if and only if $f_m(x,y)\equiv 0 \mod p$. Now we can apply the quantum version of Lucas' theorem. By Lemma \ref{lemma:quantumlucas}, $f_m(x,y)\equiv 0 \mod p$ if and only if there is an $i$ such that $f_m(x_i,y_i)\equiv 0 \mod p$, where $x_i,y_i$ are the base-$p$ digits of $x,y$. 

In general, to find the fractal generated by a quantum walk with the general coin from Equation \eqref{eq:scaledgeneralcoin} for some $n,m$ that are non-zero modulo $p$, we simply have to compute $f_m(x,y) \mod p$ for only $0\leq x,y < p$ to find what we call the \emph{base image}. Figure \ref{fig:baseplots} shows these base images for several values of $m$ and $p$.  From this the fractal can be constructed in a simple recursive way, shown in Figure \ref{fig:constructionstages}, resulting in the fractals shown in Figure \ref{fig:generalfractals}.
This recursive method is valid because each recursion step corresponds to adding another digit to $x$ and $y$, and as mentioned above, a pixel will be white if and only if there are digits (i.e. a recursion step) in which the region corresponding to those digits is white. In case of the example construction (Figure \ref{fig:constructionstages}), the third picture corresponds to $x,y$ values in the range $0\leq x,y < 2^3$ which can be described by three digits modulo 2. Let $x=[x_2 x_1 x_0]_2$ and $y=[y_2 y_1 y_0]_2$, then $x_2,y_2$ specify which of the four biggest quadrants the pixel is in. Likewise, $x_1,y_1$ specify which of the four subquadrants of that first quadrant it is in, and $x_0,y_0$ specify the final position within that subquadrant. By the construction in Figure \ref{fig:constructionstages}, the pixel will be white if and only if one of those chosen quadrants was bottom-right. This is equivalent to saying that the pixel is white if and only if there is an $i$ such that $f_m(x_i,y_i)\equiv 0\mod p$.

\begin{figure}
    \makebox[\textwidth][c]{
        \includegraphics{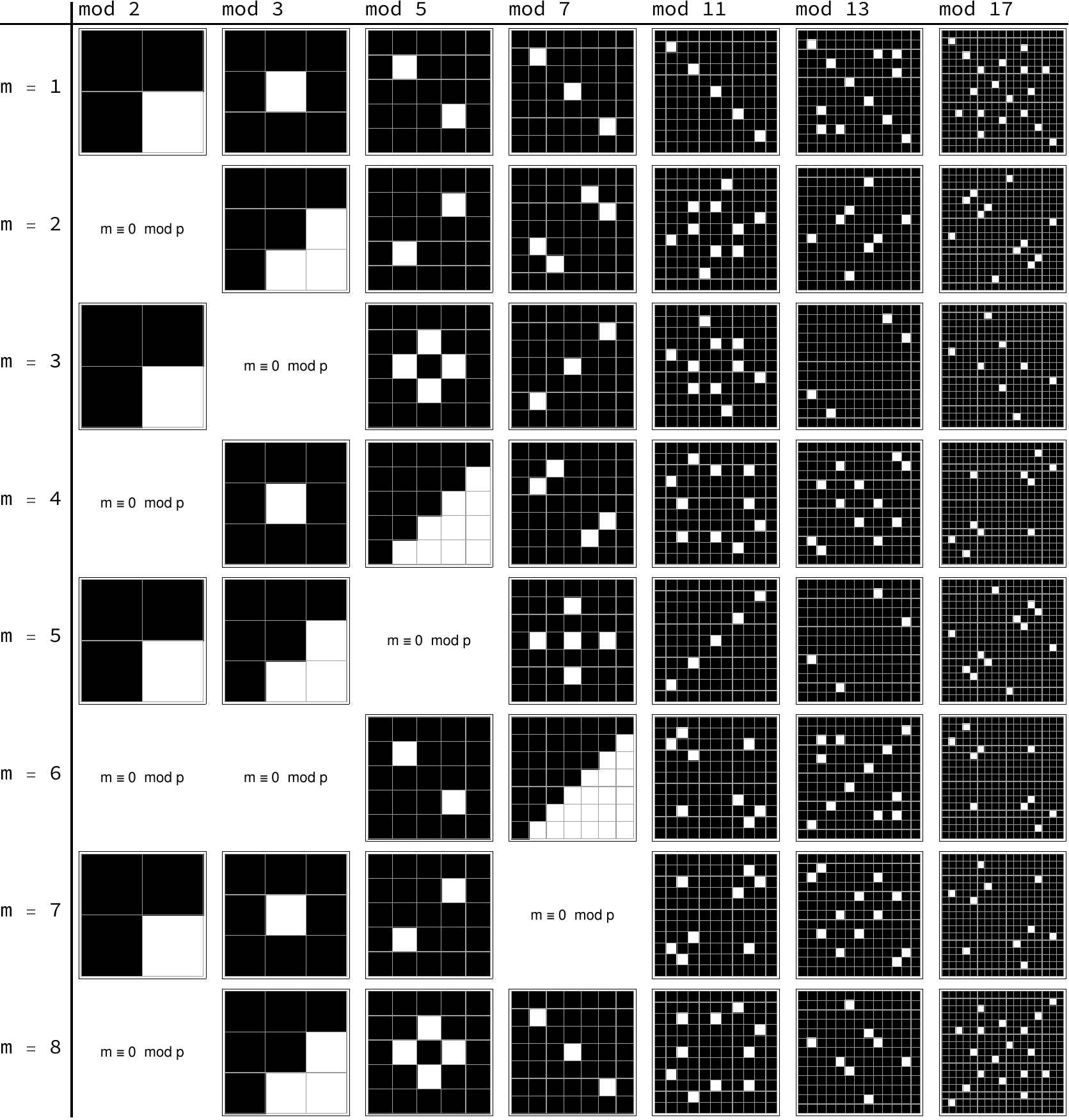}
    }
    \caption{\label{fig:baseplots} Base images: plots of $f_m(x,y) \mod p$ for $0\leq x,y < p$ for different values of $m$ and $p$ (nothing is shown for $m\equiv 0 \mod p$ for reasons explained in the text). Figure \ref{fig:constructionstages} explains how to construct the fractals from these base images and Figure \ref{fig:generalfractals} shows the resulting fractals.}
\end{figure}
\begin{figure}
    \centering
    \raisebox{-0.5\height}{\includegraphics{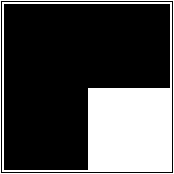}}
    $\implies$
    \raisebox{-0.5\height}{\includegraphics{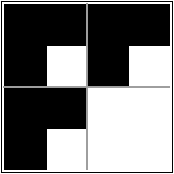}}
    $\implies$
    \raisebox{-0.5\height}{\includegraphics{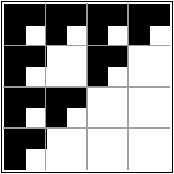}}
    $\implies$
    \raisebox{-0.5\height}{\includegraphics{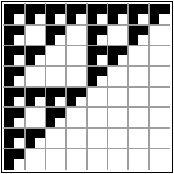}}
    \caption{\label{fig:constructionstages} Construction of the fractal from the \emph{base image}. The leftmost picture shows one of the base images from Figure \ref{fig:baseplots}. At each step, every black pixel is replaced by a copy of the base image. Infinite recursion steps yield the fractal. Some of these fractals are shown in Figure \ref{fig:generalfractals} (for finite recursion steps).}
\end{figure}
\begin{figure}
    \makebox[\textwidth][c]{
        \includegraphics[scale=0.8]{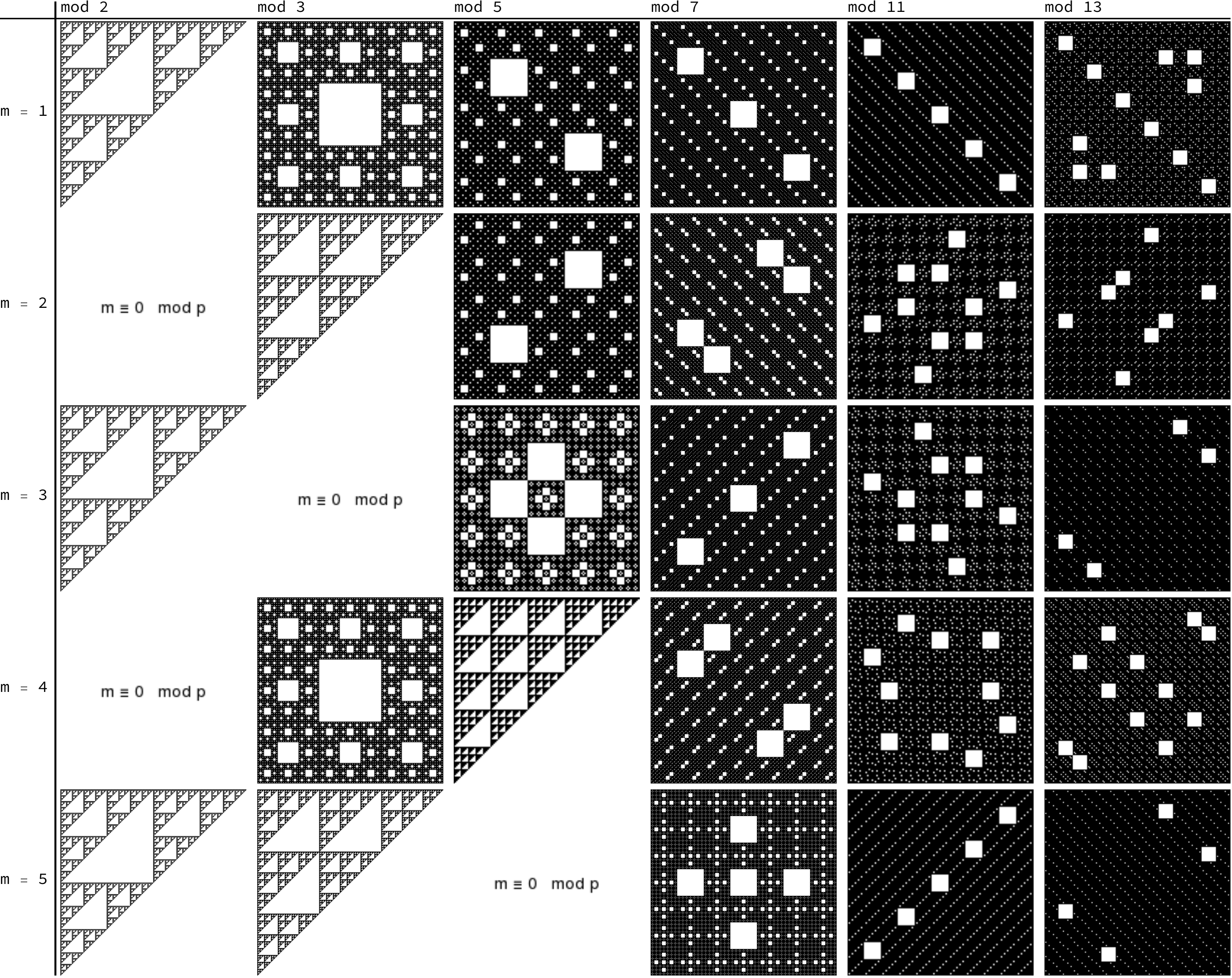}
    }
    \caption{\label{fig:generalfractals} Fractals obtained from general 1-dimensional quantum walks plotted modulo a prime. The number $m$ on the left represents the coin class, where $m=1$ includes the Hadamard coin. A pixel is coloured white if and only if the scaled probability is equal to zero modulo $p$.}
\end{figure}

\subsection{Other properties of the Hadamard triangle}
One can add the probabilities in each row of the triangle and by unitarity this sum will always be equal to one. Instead one can also consider summing all the \emph{amplitudes} in a row. Define the column vector $\Psi(t) = (\Psi_\uparrow(t) \;\;  \Psi_\downarrow(t))^T$ where $\Psi_\uparrow(t)$ is the sum of the up amplitudes at time $t$, i.e.
\begin{align*}
    \Psi_\uparrow(t) = \sum_{n=-t}^{t} \psi_\uparrow(n,t) ,
\end{align*}
and similar for $\Psi_\downarrow(t)$. Alternatively, consider the linear map
\begin{align*}
    \Sigma = \sum_{n\in\mathbb{Z}} \bra{n},
\end{align*}
so that $\Psi(t) = \Sigma \; U^t \ket{0,\uparrow}$.
Note that to go from time $t$ to $t+1$, one application of a coin and shift is performed ($U=S(\mathrm{Id}\otimes H)$), but the sums of all up or down amplitudes are invariant under the shift operation. In other words $\Sigma S = \Sigma$. Furthermore we have $\Sigma (\mathrm{Id}\otimes H) = H \Sigma$, so we have
\begin{align*}
    \Sigma U^t = \Sigma (S(\mathrm{Id}\otimes H))^t = H^t \Sigma.
\end{align*}
This can also be seen by simply looking at Figure \ref{fig:qpascal2}, and noting that $\Psi_\uparrow(t+1)=\frac{1}{\sqrt{2}}(\Psi_\uparrow(t)+\Psi_\downarrow(t))$ and $\Psi_\downarrow(t+1)=\frac{1}{\sqrt{2}}(\Psi_\uparrow(t)-\Psi_\downarrow(t))$, or simply
\begin{align*}
    \Psi(t+1) = H \Psi(t).
\end{align*}
The sum over all amplitudes, up and down, is therefore
\begin{align*}
    \Psi_\uparrow(t) + \Psi_\downarrow(t) =
    \begin{cases}
        \Psi_\uparrow(0) + \Psi_\downarrow(0) & t \text{ even}\\
        \sqrt{2} \Psi_\uparrow(0) & t \text{ odd}
    \end{cases}
\end{align*}
Note that when the process is scaled so that all numbers become integer (i.e. $H'=\sqrt{2}H$), as was done for the fractals, and the starting state is $\ket{0,\uparrow}$ then the above gives
\begin{align*}
    \Psi'_\uparrow(t) + \Psi'_\downarrow(t) =
    \begin{cases}
        2^{t/2}     & t \text{ even},\\
        2^{(t+1)/2} & t \text{ odd},
    \end{cases}
\end{align*}
so the sum of all amplitudes in a row is always a power of two.

~

Pascal's triangle has the property that summing over the so-called \emph{shallow diagonals} yields the Fibonacci sequence. The $n$'th shallow diagonal $d_n$ ($n\geq 0$) corresponds to the sum
\begin{align*}
    d_n = \sum_{c = 0}^{\lfloor n/2 \rfloor} \binom{n-c}{c},
\end{align*}
over the numbers in Pascal's triangle and is equal to the Fibonacci number $F_{n+1}$, where $F_{1}=F_{2}=1$ and $F_{n+1}=F_n + F_{n-1}$. By the property $\binom{n}{k}=\binom{n-1}{k-1}+\binom{n-1}{k}$ it follows that
\begin{align*}
    d_n &= \sum_{c\geq 1} \binom{(n-1)-c}{c-1} + \sum_{c \geq 0}\binom{(n-1)-c}{c}\\
        &= \sum_{c\geq 0} \binom{(n-2)-c}{c} + \sum_{c \geq 0}\binom{(n-1)-c}{c} = d_{n-2} + d_{n-1}.
\end{align*}
We can consider the same diagonals but now in the triangle of amplitudes of the Hadamard walk. In particular we will consider the same numbers that gave rise to the Sierpinski triangle, namely the red and blue numbers of Figure \ref{fig:qpascal2}. The blue numbers (down components) are given by $H_\mathrm{blue}$ as in Equation (\ref{eq:Hblue}). Similarly, the red numbers (up components) are given by
\begin{align*}
    H_\mathrm{red}(R,C) = \begin{cases} \sum_{k\geq 1} \binom{C+1}{k} \binom{R-C-1}{k-1} (-1)^{R-C-k} & C < R\\ 1 & C=R\end{cases}
\end{align*}
Unlike the case of Pascal's triangle, it now matters in which direction the diagonal is considered because the triangle is no longer symmetric. We therefore consider four options, corresponding to the two triangles (red and blue) and the two possible directions for the diagonals $\diagup$ and $\diagdown$. We denote the $\diagup$ diagonals by $A_\mathrm{red}$ and $A_\mathrm{blue}$ and the $\diagdown$ diagonals by $B_\mathrm{red}$ and $B_\mathrm{blue}$. They are defined as
\begin{align*}
\begin{array}{rlcrl}
    A_{\mathrm{blue},n} =& \sum_{c\geq 0} H_\mathrm{blue}(n-c,c) & ~ &
    B_{\mathrm{blue},n} =& \sum_{c\geq 0} H_\mathrm{blue}(n-c,n-2c)\\
    A_{\mathrm{red} ,n} =& \sum_{c \geq 0} H_\mathrm{red}(n-c,c) & ~ &
    B_{\mathrm{red} ,n} =& \sum_{c\geq 0} H_\mathrm{red}(n-c,n-2c) .
\end{array}
\end{align*}
Using the same property of binomial coefficients as before, we have
\begin{align*}
    A_{\mathrm{blue},n}
    &=\sum_{c\geq 0}\sum_{k\geq 0} \binom{c}{k} \binom{n-2c}{k} (-1)^{n-k} \\
    &=\sum_{c\geq 0}\sum_{k\geq 0} \binom{c}{k} \binom{n-2c-1}{k-1} (-1)^{n-k}
    + \sum_{c\geq 0}\sum_{k\geq 0} \binom{c}{k} \binom{n-2c-1}{k} (-1)^{n-k} \\
    &= A_{\mathrm{red},n-2} - A_{\mathrm{blue},n-1}
\end{align*}
Similarly we find
\begin{align*}
    A_{\mathrm{red},n} = A_{\mathrm{red},n-2} + A_{\mathrm{blue},n-1},
\end{align*}
and combining these two equations yields the same recurrence relation for both the red and blue diagonals:
\begin{align*}
    A_{n} &= - A_{n-1} + A_{n-2} + 2 A_{n-3},
\end{align*}
but with different initial conditions for the red and blue sequences.
For the diagonal in the other direction ($\diagdown$) we find
\begin{align*}
    B_{\mathrm{blue},n} &= B_{\mathrm{red},n-1} - B_{\mathrm{blue},n-2}, \\
    B_{\mathrm{red} ,n} &= B_{\mathrm{red},n-1} + B_{\mathrm{blue},n-2},
\end{align*}
which can be combined to form another recurrence relation
\begin{align*}
    B_{n} = B_{n-1} - B_{n-2} + 2 B_{n-3},
\end{align*}
that holds for both the blue and red sequence but with different initial conditions.

\section{Acknowledgements}
The authors would like to thank Florian Speelman and Jeroen Zuiddam for useful discussions and Frank den Hollander for feedback.
The work in this paper is supported by the Netherlands Organisation for Scientific Research (NWO) through Gravitation-grant NETWORKS-024.002.003.

\printbibliography

\end{document}